\documentclass[journal,draftcls,onecolumn,12pt,twoside]{IEEEtran}
 \IEEEoverridecommandlockouts
\usepackage[english]{babel}
\usepackage{fancyhdr}
\usepackage{framed}
\usepackage{amssymb}
\usepackage{latexsym}
\usepackage{amsmath}
\usepackage{psfrag}
\usepackage[T1]{fontenc}

\usepackage{epsfig}
\usepackage{graphicx}
\usepackage{color}
\newtheorem{definition}{Definition}
\newtheorem{theorem}{Theorem}
\newtheorem{proof}{Proof}

\newcommand{\PSNR}{{\sf PSNR}}

\newcommand{\Es}{\mathbb{E}_s}
\long\def\comment#1{}

\newcommand{\Ea}{\stackrel{(a)}{=}}
\newcommand{\Eb}{\stackrel{(b)}{=}}
\newcommand{\Ecc}{\stackrel{(c)}{=}}
\newcommand{\Ed}{\stackrel{(d)}{=}}
\newcommand{\Ee}{\stackrel{(e)}{=}}
\newcommand{\Ef}{\stackrel{(f)}{=}}
\newcommand{\Eg}{\stackrel{(g)}{=}}
\newcommand{\Eh}{\stackrel{(h)}{=}}
\newcommand{\prodc}{\stackrel{\cdot}{\prod}}

\newfont{\bbb}{msbm10 scaled 700}

\newfont{\bb}{msbm10 scaled 1100}
\newcommand{\CC}{\mbox{\bb C}}
\newcommand{\PP}{\mbox{\bb P}}
\newcommand{\RR}{\mbox{\bb R}}

\newcommand{\ZZ}{\mbox{\bb Z}}

\newcommand{\EE}{\mbox{\bb E}}

\newcommand{\av}{{\bf a}}
\newcommand{\bv}{{\bf b}}
\newcommand{\cv}{{\bf c}}
\newcommand{\dv}{{\bf d}}
\newcommand{\ev}{{\bf e}}
\newcommand{\fv}{{\bf f}}

\newcommand{\mv}{{\bf m}}

\newcommand{\rv}{{\bf r}}
\newcommand{\sv}{{\bf s}}
\newcommand{\tv}{{\bf t}}
\newcommand{\uv}{{\bf u}}

\newcommand{\vv}{{\bf v}}
\newcommand{\xv}{{\bf x}}
\newcommand{\yv}{{\bf y}}
\newcommand{\zv}{{\bf z}}
\newcommand{\zerov}{{\bf 0}}
\newcommand{\onev}{{\bf 1}}

\newcommand{\Fm}{{\bf F}}
\newcommand{\Gm}{{\bf G}}
\newcommand{\Hm}{{\bf H}}
\newcommand{\Id}{{\bf I}}

\newcommand{\Lm}{{\bf L}}

\newcommand{\Pm}{{\bf P}}

\newcommand{\Sm}{{\bf S}}

\newcommand{\Cc}{{\cal C}}

\newcommand{\Ec}{{\cal E}}

\newcommand{\Gc}{{\cal G}}

\newcommand{\Kc}{{\cal K}}
\newcommand{\Lc}{{\cal L}}

\newcommand{\Nc}{{\cal N}}

\newcommand{\Qc}{{\cal Q}}
\newcommand{\Rc}{{\cal R}}

\newcommand{\Uc}{{\cal U}}

\newcommand{\Xc}{{\cal X}}
\newcommand{\Yc}{{\cal Y}}

\newcommand{\hu}{{\underline h}}
\newcommand{\gu}{{\underline g}}
\newcommand{\N}{{\mathrm N}}

\newcommand{\deltav}{\hbox{\boldmath$\delta$}}

\newcommand{\eqdef}{\stackrel{\Delta}{=}}

\newcommand{\GF}{{\sf GF}}
\newcommand{\LLR}{{\sf LLR}}

\newcommand{\x}{{\sf x}}
\newcommand{\y}{{\sf y}}
\newcommand{\p}{{\sf p}}
\newcommand{\rr}{{\sf r}}
\newcommand{\X}{{\sf X}}
\newcommand{\Y}{{\sf Y}}

\begin{document}

\title{Joint Source-Channel Coding for Deep-Space Image Transmission using Rateless Codes}

\author{{\normalsize O. Y. Bursalioglu, G. Caire, and D. Divsalar}}
\maketitle
{\let\thefootnote\relax\footnotetext{
\noindent
O. Y. Bursalioglu is with Docomo Innovations Inc., Palo Alto, CA. G. Caire are with the Ming Hsieh Dept. of Electrical Engineering, University of Southern
California, Los Angeles, CA. D. Divsalar is with Jet Propulsion Laboratory, California Institute of Technology, Pasadena, CA.

\thanks{This research in part was carried out at the Jet Propulsion Laboratory, California Institute of Technology, under a contract with NASA. The work by the University of Southern California, and JPL was funded through the NASA/JPL/DRDF/SURP Program.}
}}

\begin{abstract}
A new coding scheme for image transmission over noisy channel is proposed.
Similar to standard image compression, the scheme includes a linear transform followed by successive refinement
scalar quantization.
Unlike conventional schemes, in the proposed system the quantized transform coefficients are linearly mapped
into channel symbols using systematic linear encoders. This fixed-to-fixed length ``linear index coding'' approach avoids
the use of an explicit entropy coding stage (e.g., arithmetic or Huffman coding), which is typically fragile
to channel post-decoding residual errors.  We use linear codes over $\GF(4)$, which are particularly suited for this application,
since they are matched to the dead-zone quantizer symbol alphabet and to the QPSK modulation used
on the deep-space communication channel. We optimize the proposed system where the linear codes are
systematic Raptor codes over $\GF(4)$. The rateless property of Raptor encoders allows to achieve a ``continuum'' of coding rates,
in order to accurately match the channel coding rate to the transmission channel capacity and to the
quantized source entropy rate for each transform subband and refinement level.
Comparisons are provided with respect to the concatenation of state-of-the-art image coding and channel coding
schemes used by Jet Propulsion Laboratories (JPL) for the Mars Exploration Rover (MER) Mission.

\end{abstract}

\newpage

\section{Introduction}\label{sec:introduction}

In conventional digital image transmission over noisy channels, the source coding and channel
coding stages are designed and operated separately. Image coding is usually implemented by a linear transformation (e.g., DCT, Wavelet),
followed by the transform coefficients quantization and by entropy coding of the resulting quantization bits.
Due to the lack of robustness of standard entropy coding schemes, a few bit errors after the channel decoder
may dramatically corrupt the decoded image. To prevent this catastrophic error propagation, the source is partitioned into segments,
such that the effect of errors is spatially confined. In order to preserve integrity, which is a strict requirement in
deep-space scientific missions, the segments affected by errors are retransmitted at the cost of significant
delay and power expenditure. Because of the sharp waterfall behavior of the Bit-Error Rate (BER) of the powerful channel coding
schemes used in deep-space communications, slight changes in the transmission channel quality (e.g., SNR fluctuations
due to atmospheric conditions or antenna misalignment) result in dramatic degradation of the post-decoding
BER, producing sequences of highly corrupted segments that need retransmission \cite{kiely-icer-report}.

In this paper we consider the application of the Joint Source Channel Coding (JSCC) scheme
developed in f{Ozgun-Maria-JSCC-08}, \cite{BSBC} to the specific problem of deep-space image transmission.
The proposed JSCC scheme consists of a successive refinement (also referred to as ``embedded'') quantizer,
and a family of linear codes that directly map the sequences of quantization symbols generated at each refinement
level into channel codewords. This approach is referred to as {\em Quantization with Linear Index Coding}
(QLIC). The linear mapping of the redundant quantization symbols into channel-encoded symbols replaces the non-linear entropy coding stage of conventional
source encoders. QLIC can achieve the same (optimal) entropy compression rate of conventional entropy encoders,
but it is much better conditioned in terms of residual error propagation. Similar to JPEG2000 \cite{TaMa}, we apply Discrete Wavelet Transform to the
image and then quantize the transform coefficients using a dead-zone quantizer. Since an embedded dead-zone quantizer divides
the quantization cells into at most three regions at every refinement level, the quantization indices are naturally represented as non-binary symbols.
Differently from our previous work in \cite{Ozgun-Maria-JSCC-08}, \cite{BSBC},  here we use nonbinary Raptor codes
(notably, over $\GF(4)$) for QLIC. We prove an ``isomorphism'' between the original source-channel coding problem and a ``virtual''
purely channel coding problem where the source symbols are sent through an appropriate discrete symmetric memoryless channel
over $\GF(4)$ and the channel-coded symbols are sent through the AWGN channel with QPSK modulation, which is the standard
in deep-space communications. This isomorphism allows us to cast the non-standard code optimization in the source-channel
coding case as  a more familiar optimization for the purely channel coding case,
which we solve by using a modified EXIT chart technique~\cite{EtSho06}.

The three components the proposed JSCC scheme, namely a wavelet transform, a scalar embedded quantizer,
and a linear encoding stage, are examined in Sections \ref{sec:wavelet}, \ref{sec:quantizer} and \ref{sec:code-design}, respectively.
Section \ref{sec:system-setup} introduces the notation used throughout the paper and defines the relevant system
optimization problem for JSCC based on the concatenation of embedded quantization and channel coding in general.
In Sec. \ref{sec:Results}\footnote{These results appeared previously in \cite{JPL-ITA} as a conference proceeding. This work includes more details about the scheme and derivations. }, we compare the performance of the proposed scheme with the state-of-the art image transmission scheme
for deep-space communication channel. This {\em baseline scheme} is based on the separation of source compression and channel coding.
Our results show that when the channel quality is perfectly known, the highly optimized baseline scheme
provides {\em slightly} higher efficiency. However, as soon as the channel conditions degrade, the proposed JSCC scheme
offers significant {\em robustness} advantages. In particular, it is able to handle fluctuations of the channel SNR as large as 1 dB below its nominal value,
with visually acceptable quality and without requiring retransmissions.

\section{System Setup}\label{sec:system-setup}

The deep-space transmission channel is represented by the
discrete-time complex baseband equivalent model
\begin{equation} \label{awgn}
y_t = \mu(x_t) + z_t, \;\;\; t = 1,2,\ldots,
\end{equation}
where $y_t\in\CC$, $x_t \in \GF(q)$ is a coded symbol taking on values in a finite field,
$\mu : \GF(q) \rightarrow \mathfrak{X}$ is a {\em labeling map} of a signal constellation
$\mathfrak{X} = \{\Xc_0,\ldots, \Xc_{q-1}\}$ with the elements of $\GF(q)$ and
$z_t \sim \Cc\Nc(0,N_0)$ is the complex circularly symmetric AWGN.
The channel signal-to-noise ratio (SNR) is given by $E_s/N_0$, where $E_s = \frac{1}{q} \sum_{j=0}^{q-1} |\Xc_j|^2$ is the average power of the
signal constellation. We indicate by $C_{\mathfrak{X}}(E_s/N_0)$ the maximum achievable rate of channel (\ref{awgn}) when the input
$x_t$ is i.i.d. and uniformly distributed over $\GF(q)$.~\footnote{We shall refer to $C_{\mathfrak{X}}(E_s/N_0)$ as ``channel capacity'' even though, for general
constellations, the uniform input probability may not be capacity achieving. As a matter of fact, for the case of QPSK considered in the rest of the paper
the uniform input probability does achieve capacity.}

A source block of length $K$ is denoted by  $\Sm \in \RR^{s\times K}$, where
$\Sm(i,:) = (S(i,1), \ldots, S(i,K))$ is the $i$-th row of $\Sm$, with variance $\sigma_i^2 \eqdef \frac{1}{K} \EE[\|\Sm(i,:)\|^2]$,
is referred to as referred to as the $i$-th {\em source component}.
A $(s \times K)$-to-$N$ source-channel code for source $\Sm$ and channel (\ref{awgn})
is formed by an encoding function $\Sm \mapsto \xv = (x_1,\ldots,x_N)$, and by a
decoding function $\yv = (y_1, \ldots, y_N) \mapsto \widehat{\Sm}$.

Letting $d_{i} = \small{\frac{1}{K} } \EE[ \|\Sm(i,:) - \widehat{\Sm}(i,:) \|^2]$ denote the mean-square error
for the $i$-th component, the weighted mean-square error (WMSE) is defined by
\begin{equation} \label{eq:D_l_Dist_weighted}
D = \small{\frac{1}{s}} \sum_{i=1}^s v_i d_{i},
\end{equation}
where $\{v_i\}$ is a set of non-negative weights that depends on the specific application (see Section \ref{sec:wavelet}).
Let $r_i(\cdot)$ denote the rate-distortion (R-D) function of the $i^{\rm th}$ source component
with respect to the MSE distortion. Then the R-D function of $\Sm$ with respect to the WMSE distortion is given by
\begin{equation} \label{gen-rev-waterfilling}
\Rc(D) =  \min \; \frac{1}{s} \sum_{i=1}^s r_i(d_i),\;\; \mbox{subject to} \; \frac{1}{s} \sum_{i=1}^s v_i d_i = D,
\end{equation}
where the optimization is with respect to the values $d_i \geq 0$ for $i = 1,\ldots, s$.
For example, for parallel Gaussian sources and equal weights ($v_i = 1$ for all $i$),
(\ref{gen-rev-waterfilling}) yields the well-known ``reverse waterfilling'' formula (see \cite[Theorem 10.3.3]{CoTh}).
For a family of successive refinement source codes with R-D functions $r_i(d)$, $i = 1,\ldots, s$,
assumed to be convex, non-increasing \cite{ortega-ramchandran-RD}
and identically zero for $d > \sigma_i^2$, the operational R-D function of the source $\Sm$ is also given by (\ref{gen-rev-waterfilling}).
Therefore, in the following, $\Rc(D)$ is used to denote the actual operational
R-D function of for some specific, possibly suboptimal,  successive refinement source code.

We define the source-channel bandwidth efficiency of the encoder $\Sm \mapsto \xv$ as the ratio
$b = \frac{N}{sK}$, measured in channel uses per source sample. This corresponds to the familiar notion of ``bit per pixel''
in the case where the source symbols are pixels (image coding) and the channel is just a storage device for which one channel use
corresponds to storing one bit. By analogy, in this paper $b$ will expressed in ``symbol per pixel'' (spp).
It is immediate from the definition of R-D function that the minimum distortion $D$ achievable at channel capacity
$C_{\mathfrak{X}}(E_s/N_0)$ and source-channel bandwidth efficiency $b$ is given by $D = \Rc^{-1}(bC_{\mathfrak{X}}(E_s/N_0))$.
\section{Subband Coding}\label{sec:wavelet}

Images are decomposed into a set of source components
by a Discrete Wavelet Transform (DWT). In this work we make use of the DWT described in
JPEG2000  \cite{TaMa} for lossy compression. With $W$ levels of DWT, the transformed image is partitioned
into $3W+1$ ``subbands''. A subband decomposition example is given in Fig.~\ref{fig:subsources}-A for $W=3$.
This produces $3W+1 = 10$ subbands, which in the figure are indicated by
LL0, HL1, LH1, HH1, HL2, LH2, HH2, HL3, LH3, HH3, respectively.
The subbands have different lengths, all multiples of the LL0 subband length.
For simplicity, we partition the DWT into source components of the same length, all equal to the
the length of the LL0 subband. This yields $s = 2^{2W}$ source component blocks of
length $K = \Kc^2/s$, where $\Kc \times \Kc$ indicates the size of the original image in pixels.

Since this DWT is a bi-orthogonal transform, the MSE distortion in the pixel domain
is not equal to the MSE distortion in the wavelet domain.
In our case, for $W=3$, the weight of a source component block in subband $w = \{1,\ldots, 10\}$ is given by the
$w$-th coefficient of the vector $[l^6,\, l^5h,\, l^5h,\,l^4h^2,\, l^3h^3,\, l^3h^3,\,l^2h^2,\, lh,\,lh,\, h^2]$,
where, for the particular DWT considered (namely, the CDF 9/7~\cite{Daubechies} wavelet),
we have $l=1.96$ and $h=2.08$ \cite{TaMa}.
\begin{figure}[ht]
\centerline{(A) \includegraphics[width=5.0cm,height=5.0cm]{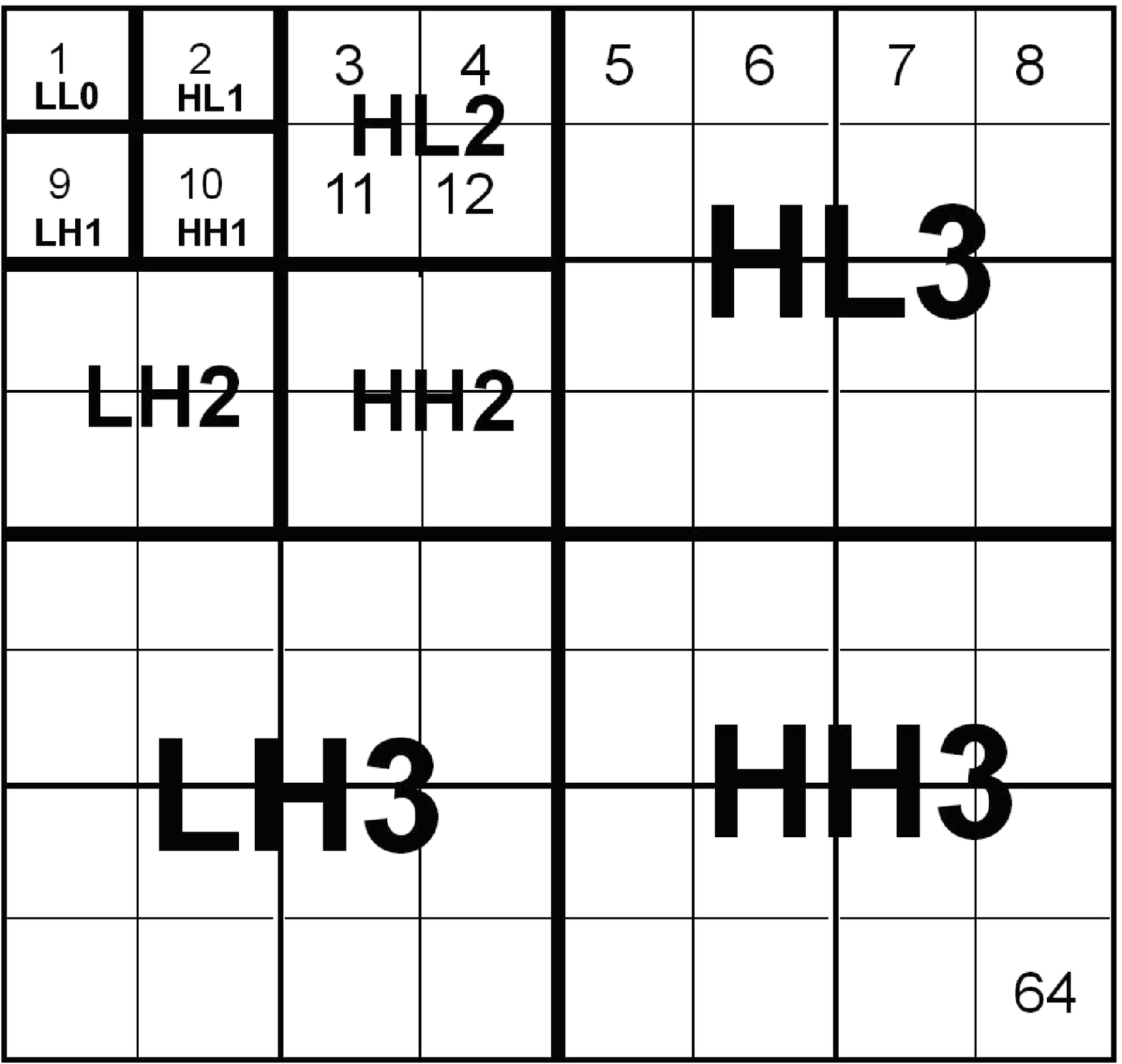} \hspace{1cm}
(B)\includegraphics[width=6cm,height=5.0cm]{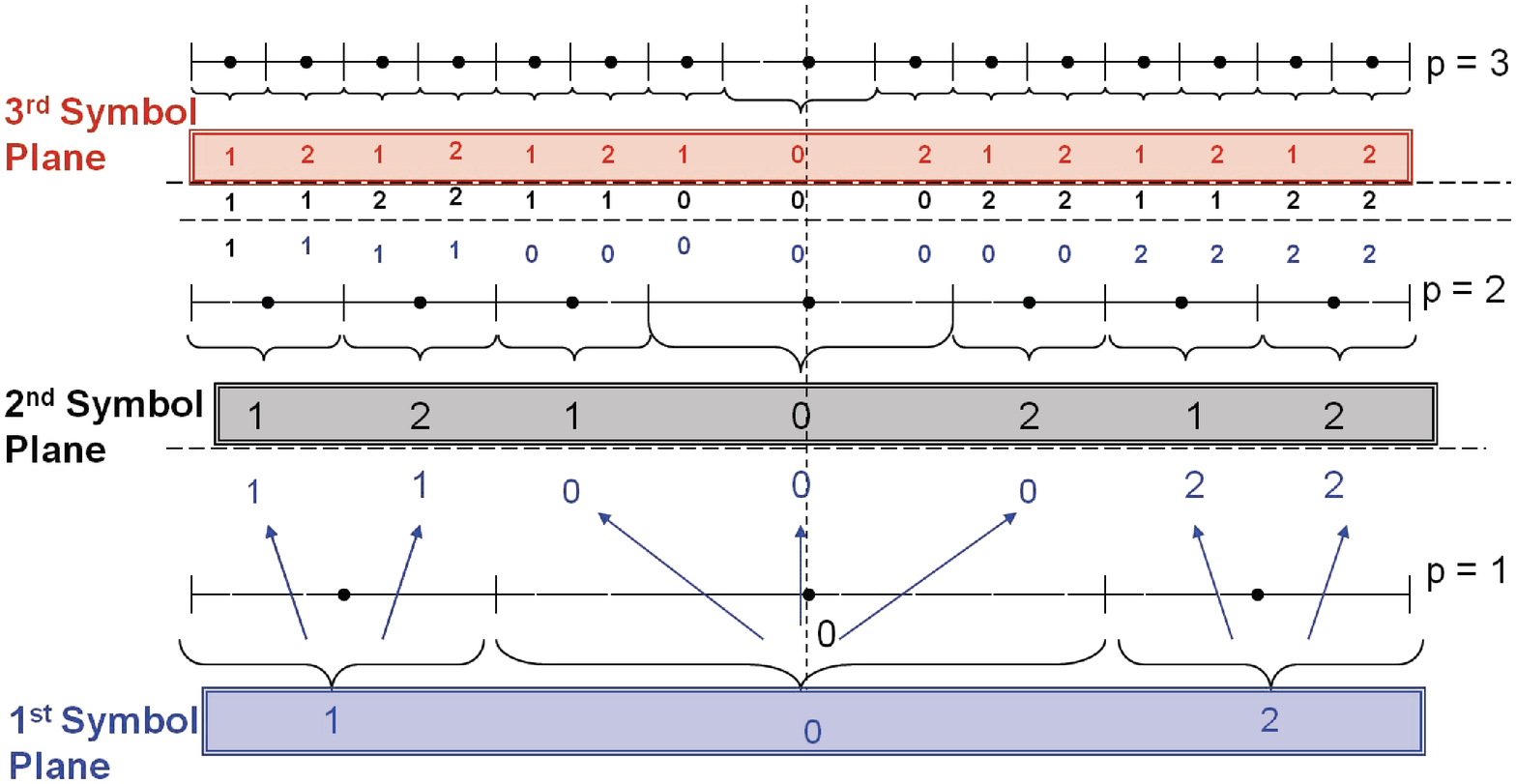}}
\caption{(A): $W = 3$, partitioning of an image into $10$ subbands and $64$ source components. (B): Quantization cell indexing for a
embedded dead-zone quantizer with $p=1,2,3$. }
\label{fig:subsources}
\end{figure}
The subband LL0 consists approximatetely of a decimated version
of the original image.
In order to obtain better compression in the transform domain,
a Discrete Cosine Transform (DCT) is applied to subband LL0 so that its energy
is ``packed'' into a very few coefficients The resulting few high-energy coefficients are separately encoded and transmitted as part of the header. This is highly protected by a sufficiently low rate channel code and is not discussed further in this work since it has a negligible contribution to the overall coding length.
After extracting these few high-energy coefficients,
all subbands show similar marginal statistics, well-suited for the
embedded dead-zone quantizer described in the next section.
\section{Embedded Scalar Quantization}\label{sec:quantizer}

The simplest form of quantization defined in JPEG2000 is a uniform scalar quantizer
where the center cell's width is twice the width of the other cells, at any resolution level.
For example, Fig.~\ref{fig:subsources}-B shows  such a quantizer with $3$ resolution levels.
This scheme, referred to as ``dead-zone'' quantizer, is adopted in this work.
We indicate the cell partition at every level by symbols $\{0,1,2\}$ as shown in Fig.\ref{fig:subsources}-B.
The scalar quantization function is denoted by $\Qc : \RR \rightarrow \{0,1,2\}^P$,
where $2^{P+1}-1$ is the number of quantization regions for the highest level of refinement.
Let $\uv^{(i)} = \Qc(\Sm(i,:))$ denote the block of ternary quantization indices, formatted as a $P \times K$ array.
The $p$-th row of $\uv^{(i)}$, denoted by $\uv^{(i)}(p,:)$, is referred to as the
$p$-th ``symbol-plane'', where $\uv^{(i)}(1,:)$ corresponds to the coarser refinement
and $\uv^{(i)}(P,:)$ to the finest. A refinement level $p$ consists of all symbol planes from 1 to $p$.
The quantization distortion for the $i$-th source component at refinement level $p$ is denoted  by $D_{\Qc}(i,p)$.


The quantizer output $\uv^{(i)}$ can be considered as a {\em discrete memoryless source},
with entropy rate $H^{(i)} = \frac{1}{K} H(\uv^{(i)})$ (in bits/source symbol).
The chain rule of entropy \cite{CoTh} yields $H^{(i)} = \sum_{p=1}^P H_p^{(i)}$, with
$
H_p^{(i)} = \frac{1}{K} H \left ( \left . \uv^{(i)}(p,:) \right | \uv^{(i)}(1,:),\ldots, \uv^{(i)}(p-1,:) \right ), \;\;\; p = 1,\ldots,P.
$
Then,  the set of R-D points achievable by the concatenation of the quantizer using
$0,1,\ldots, P$ quantization levels~\footnote{Notice: 0 quantization levels indicates that the whole source component is
reconstructed at its mean value.} and ideal entropy coding is given by
\begin{equation} \label{quant-achievable}
\left ( \sum_{j=1}^p H_j^{(i)} , D_{\Qc}(i,p) \right ), \;\;\; p = 0,\ldots,P,
\end{equation}
where, by definition, $D_{\Qc}(i,0) = \sigma_i^2$. Using time-sharing, any point in the convex hull of the above achievable
points is also achievable. Therefore, the operational R-D curve $r_i(d)$ of the scalar quantizer is
given by the {\em lower convex envelope} of the points in (\ref{quant-achievable}).

By construction, $r_i(d)$ is piecewise linear, convex and decreasing on the domain $D_{\Qc}(i,P) \leq d \leq \sigma_i^2$.
As such, it is possible to represent $r_i(d)$ as the pointwise maximum of the family of straight lines
joining the pairs of R-D points in (5), for consecutive indices $p$ and $p+1$.
Using this observation in (\ref{gen-rev-waterfilling}),  the minimum WMSE distortion
with capacity $C_{\mathfrak{X}}(E_s/N_0)$ and bandwidth efficiency $b$ is the result of the linear program:
\begin{eqnarray} \label{eq:MWTD-linear}
&&\mbox{minimize} \;\;  \frac{1}{s} \sum_{i=1}^s v_i d_{i} \\
&&\mbox{subject to} \;\; \frac{1}{s}\sum_{i=1}^s\gamma_i\leq bC_{\mathfrak{X}}(E_s/N_0);\;
 D_{\Qc}(i,P) \leq d_i \leq \sigma_i^2,\;\forall i;\; \gamma_i \geq a_{i,p} d_i + b_{i,p},\;\forall i,p,\nonumber
\end{eqnarray}
where $a_{i,p} d + b_{i,p}$ is the $p$-th straight line (for appropriate coefficients ($a_{i,p}, b_{i,p}$) obtained by linear interpolation of the
points in (\ref{quant-achievable})) forming $r_i(d)$ as said before.~\footnote{The details of linear interpolation are trivial and are omitted for the sake of brevity.}

While (\ref{eq:MWTD-linear}) assumes a capacity achieving channel code,
in the proposed JSCC scheme the refinement levels (symbol planes)
of each source component are encoded by actual codes of finite block length.
Letting $n_p^{(i)}$ denote the number of channel encoded symbols for the $p^{\rm th}$ plane of the $i^{\rm th}$ source
component, the total channel coding block length is given by $N = \sum_{i=1}^s \sum_{p=1}^P n_{p}^{(i)}$.
Consistent with the definition of the Raptor code {\em overhead} for channel coding applications \cite{EtSho06}, we define the
overhead $\theta_p^{(i)}$ for JSCC such that $n_p^{(i)} = \frac{K H_p^{(i)}(1+\theta_{p}^{(i)})}{C_{\mathfrak{X}}(E_s/N_0)},$
where $K H_p^{(i)}/C_{\mathfrak{X}}(E_s/N_0)$ is the information theoretic lower bound to the block length, obtained
from the source-channel coding converse theorem \cite[Theorem 8.13.1]{CoTh}.

In the case of a family of practical codes characterized by their overhead coefficients $\{\theta_p^{(i)}\}$, the
computation of the achievable R-D function takes on the same form of (\ref{eq:MWTD-linear}), where the
coefficients $\{a_{i,p}, b_{i,p}\}$ are obtained from the linear interpolation of the modified R-D points
\begin{equation} \label{quant-achievable-mod}
\left ( \sum_{j=1}^p H_j^{(i)}(1 + \theta_j^{(i)})  , D_{\Qc}(i,p) \right ), \;\;\; p = 0,\ldots,P.
\end{equation}
(see  \cite{BSBC} for details).  For given code families and block lengths, the overhead factors $\theta_p^{(i)}$ can be
experimentally determined, and used in the system optimization.

To give an idea of the symbol plane entropies resulting from deep-space images, in (\ref{table:entropy}) we give
such values for the first source component (subband LL0 after DCT) of a test image from the Mars Exploration Rover,
which will be referred to in the following as image MER1:
\begin{equation}\label{table:entropy}
 \left[H_{1}^{(1)},\ldots,H_{8}^{(1)} \right]=  \left[0.0562,\, 0.0825,\,0.2147,\,0.4453,\,0.8639,\,1.1872,\,1.1917,\,1.1118\right].
\end{equation}
By examining a large library of such images, we observed
that the range of values shown in (\ref{table:entropy}) are typical for this application.
\vspace{-0.5cm}
\section{Channel Coding Optimization}\label{sec:code-design}

In this section we discuss the optimization of the linear channel coding stage.
For simplicity, we focus on a single discrete source $\uv \in \GF(q)^K$
with entropy $H$, to be transmitted over the AWGN channel (\ref{awgn}) with capacity $C_{\mathfrak{X}}(E_s/N_0)$.
Obviously, the optimization procedure devised here can be applied to each source component and quantization
layer pair $(i,p)$, by letting $H = H_p^{(i)}$ and block length $n = n_p^{(i)}$.

Linear source codes are known to achieve the entropy rate of memoryless sources \cite{Csizar81}.
Several works have considered entropy-achieving fixed-to-fixed linear coding for ``almost-lossless'' data
compression \cite{Ancheta}, \cite{Massey}, \cite{CVS_itw03}, \cite{CVS_ldpc} and \cite{CVSS_fountain}.
Linear data compression codes can be directly obtained from linear error correcting codes
originally designed for additive-noise discrete memoryless channels. This is due to the following
fact \cite{CVS_ldpc}. Consider a linear fixed length data compression code given by a $K \times n$ matrix $\Hm$ which maps the source vector
$\uv$ (of length $K$) to the compressed vector $\cv = \uv\Hm$. The optimal {\em maximum a posteriori} (MAP) decoder
selects $\widehat{\uv}$ to be the most likely source vector satisfying $\widehat{\uv} \Hm = \cv$.
Next, consider a discrete additive noise channel $\yv = \xv + \uv$,
where the source $\uv$ acts as the additive noise. Let $\xv$ be a codeword of the a linear code
with parity-check matrix $\Hm$. The MAP decoder
in this case  computes the syndrome $\cv = \yv \Hm = \uv \Hm$ and finds $\widehat{\uv}$ to be the most likely noise realization
satisfying the syndrome equation $\widehat{\uv} \Hm = \cv$.  Then, it obtains the MAP decoded codeword as $\widehat{\xv} = \yv - \widehat{\uv}$.
It is clear that the optimal decoder for the data compression problem is identical to the optimal decoder for the channel coding problem. Therefore,
the achieved block error rates are identical. As a consequence, if $\Hm$ denotes a sequence (for increasing $K$)
of capacity achieving parity-check matrices for the discrete additive noise channel $\yv = \xv + \uv$, then the {\em same} sequence
of matrices achieves the entropy of the source $\uv$. In fact, in this case channel capacity and source entropy are related by $C = \log q - H$.

In order to extend the above argument from pure data compression to the transmission of compressed data over a noisy channel
it is sufficient to concatenate two linear encoding stages,  $\cv = \uv \Hm_{\rm comp}$ for data compression, and $\xv = \cv \Gm_{\rm cod}$
for channel coding.
Since the concatenation of two linear maps is a linear map,  by optimizing over all linear
maps (not necessarily decomposed as the product $\Hm_{\rm comp} \Gm_{\rm cod}$), it follows that there must exist good linear
joint source-channel codes. From now on, we shall indicate this single encoding map
by $\xv = \uv \Hm$. This can also be interpreted as {\em systematic encoding} followed by puncturing of the source symbols.
Encoding with the systematic generator matrix $\Gm = [\Id, \Hm]$ yields the systematic codeword $[\uv, \xv = \uv\Hm]$.
Then, the source symbols $\uv$ are completely punctured and only $\xv$ is transmitted. This approach is meaningful
from an information theoretic viewpoint since, in the limit of large
block length, any scheme transmitting the (redundant) source symbols directly over the channel is necessarily bounded away from
capacity. In fact, the source symbols are non-uniformly distributed with entropy $H < \log q$ and therefore
do not follow the capacity-achieving distribution.
Viewing the encoding map as systematic encoding followed by puncturing
will be instrumental to the proposed use of systematic Raptor codes for this problem, as discussed later on.

Up to this point we assumed that the noisy channel is also additive over $\GF(q)$, and therefore it is ``matched''
to the source alphabet, so that linearity can be defined.
However, in the case of deep-space transmission, the channel (\ref{awgn})
is defined over the complex field, and the codeword $\xv$ is mapped onto a sequence of modulation symbols by
the labeling map $\mu$. In order to carry over the previous arguments to this case we need a ``matching condition'' between the additive
group of the source alphabet $\GF(q)$ and an isometry group induced on the signal constellation.
For this purpose, we consider {\em geometrically uniform} constellations as defined in \cite{Forney-geo-uniform}.

\begin{definition} \label{geom-unif}
A signal set $\mathfrak{X}$ is called {\em geometrically uniform} if, given any two points $\Xc_a, \Xc_b \in \mathfrak{X}$,
there exists an isometry $w_{a,b} : \CC \rightarrow \CC$ that maps $\Xc_a$ into $\Xc_b$ while
leaving $\mathfrak{X}$ invariant.
\hfill $\lozenge$
\end{definition}

The set of all isometries that leave $\mathfrak{X}$ invariant forms the symmetry group of $\mathfrak{X}$, under the operation of
mapping composition. A subgroup $\Gc(\mathfrak{X})$ of the symmetry group
of minimal size able to generate the whole constellation $\mathfrak{X}$ as the orbit of any of its points is called a {\em generating group} \cite{Forney-geo-uniform}.
By definition, $|\Gc(\mathfrak{X})| = |\mathfrak{X}| = q$. Given an initial point $\Xc_0 \in \mathfrak{X}$, we have
$\mathfrak{X} = \{w(\Xc_0) : w\in\Gc(\mathfrak{X})\}$. This induces a one-to-one mapping  $\mu:\Gc(\mathfrak{X})\rightarrow\mathfrak{X}$ referred
to as an {\em isometric labeling}. The isometric labeling $\mu$ induces a group structure on $\mathfrak{X}$. At this point, the sought ``matching'' condition
can be stated as follows: we let  $\mathfrak{X}$ be a geometrically uniform signal constellation  admitting a generating group $\Gc(\mathfrak{X})$
isomorphic to the additive group of $\GF(q)$.

For example, a $q$-PSK signal constellation is geometrically uniform, and admits a generating group formed by the set of rotations
of multiples of $\frac{2\pi}{q}$.  This group is isomorphic to the additive group of $\ZZ_q$ (integers modulo $q$).
For $q$ prime, the ring $\ZZ_q$ coincides with the field $\GF(q)$, therefore the generating group of the $q$-PSK constellation
is isomorphic to the additive group of $\GF(q)$. A possible approach for our code design considers $q$ prime and uses $q$-PSK as
constellation. In particular, $q = 3$ is sufficient to represent the dead-zone quantizer symbols and the 3-PSK generating group
consists of the rotations $\Gc = \{I, R, R^2\}$ where $R$ is a $\pi/3$ rotation in $\CC$.

Another example is provided by $q = 4$. The additive group of $\GF(4)=\{(0,0), (1,0), (0,1), (1,1)\}$ (binary vectors of length 2, with modulo 2 addition)
is isomorphic to the additive group of  $\GF(2)\times \GF(2)$. This group is isomorphic to the isometry group formed by
$\Gc = \{I, R_x, R_y, R_{xy}\}$, where $I$ is identity, $R_x$ is reflection with respect to the real axis, $R_y$ is reflection with respect to the imaginary axis,
and $R_{xy} = R_x R_y$ is reflection with respect to the origin. Note that the isometric labeling of the 4-PSK constellation by the elements
of $\GF(4)$ coincides with the well-known Gray Mapping, which is routinely used in deep-space communications.

With the above conditions on $\mathfrak{X}$ and its isometric labeling $\mu$, we introduce the following notation:
$w_x \in \Gc(\mathfrak{X})$ denotes the isometry such that $w_x(\mu(0)) = \mu(x)$;
$w_{\xv}$ for a sequence $\xv \in \GF(q)^n$ denotes the sequence of isometries $w_{x_t}$ for $t = 1,\ldots, n$;
$\mu(\xv)$ indicates the sequence of constellation points $\mu(x_t)$ for $t = 1,\ldots, n$.

We wish to translate the non-conventional source-channel coding problem at hand into a channel coding problem defined
over a particular channel, that we refer to as the associated {\em two-block composite channel}.
We do so in order to reuse known techniques for optimizing the linear encoding
matrix $\Hm$ for the associated channel coding problem.

\begin{definition} \label{composite-channel}
The associated two-block composite channel is a $q$-ary input channel where the input is divided into two blocks,
indicated by $\vv$ and $\cv$, of length $K$ and $n$, respectively.
The first block is sent through the discrete additive noise channel defined by $\sv = \vv - \uv$, where operations are over $\GF(q)$
and where $\uv$ has the same statistics of the source. The second block is sent through the $q$-ary AWGN channel
defined by $\rv = \mu(\cv) + \zv$, where $\zv \sim \Cc\Nc(\zerov,N_0\Id)$, as in the original AWGN channel (\ref{awgn}).
\hfill $\lozenge$
\end{definition}

For the associated two-block composite channel, we consider the systematic encoder $[\vv, \cv = \vv\Hm]$. Then, we have:

\begin{theorem} \label{map-isomorphism}
The source-channel coding scheme with source $\uv$, linear encoder $\xv = \uv \Hm$, transmission over the noisy channel
$\yv = \mu(\xv) + \zv$, and MAP decoding, is {\em equivalent} to a channel coding scheme
over the associated two-block composite channel with systematic encoding and MAP decoding,
in the sense that the error region of the source-channel MAP decoder of the former
is congruent (via an isometric transformation) to the error region of the MAP decoder
of the latter, for any source vector $\uv$ and transmitted information vector $\vv$.
The isometric transformation of the two error regions depends, in general, on $\uv$ and $\vv$.
\end{theorem}

\begin{proof}
Since the two-block composite channel is symmetric by construction, and the systematic code $[\vv, \cv = \vv\Hm]$ is linear, it is immediate
to show that the MAP decoding error regions for different codewords are mutually congruent. Hence without loss of generality,
it is sufficient to consider $\vv = \zerov$, yielding the all-zero codeword. The MAP decoder for the source-channel coding scheme is
given by
\begin{equation} \label{map-1}
\widehat{\uv} = {\rm arg} \; \max_{\uv': \xv' = \uv' \Hm} \;\; \exp\left ( - \frac{1}{N_0} \left \| \yv - \mu(\xv') \right \|^2\right )  P_U(\uv').
\end{equation}
The MAP decoder for the two-block composite channel coding scheme is given by
\begin{equation} \label{map-2}
\widehat{\vv} = {\rm arg} \; \max_{\vv' : \cv' = \vv' \Hm} \;\;  \exp\left ( - \frac{1}{N_0} \left \| \rv - \mu(\cv') \right \|^2\right )  P_U(\vv' - \sv).
\end{equation}
Using the properties of the geometrically uniform constellation $\mathfrak{X}$, the generic term in the maximization of (\ref{map-1}) can be written as
\begin{eqnarray} \label{zio2}
\exp\left (  \frac{\left \| \yv - \mu(\xv') \right \|^2}{-N_0} \right )  P_U(\uv')
& = &   \exp\left (  \frac{\left \| \mu(\xv) + \zv  - \mu(\xv') \right \|^2}{-N_0} \right )  P_U(\uv') \nonumber \\
& = &   \exp\left (  \frac{\left \| w_{\xv} (\mu(\zerov)) + \zv  - \mu(\xv') \right \|^2}{-N_0} \right )  P_U(\uv') \nonumber \\
& = &   \exp\left (  \frac {\left \| \mu(\zerov) + w_{-\xv}  (\zv)  - w_{-\xv} (\mu(\xv')) \right \|^2}{-N_0}\right )  P_U(\uv') \nonumber \\
& = &   \exp\left (  \frac{\left \| \mu(\zerov) + w_{-\xv}  (\zv)  - \mu(\xv' - \xv) \right \|^2}{-N_0} \right )  P_U(\uv').
\end{eqnarray}
When $\vv = \zerov$ is transmitted, the generic term in the maximization of (\ref{map-2}) becomes
\begin{eqnarray} \label{zio1}
\exp\left (  \frac{\left \| \mu(\zerov) + \zv - \mu(\cv') \right \|^2}{-N_0} \right )  P_U(\vv' + \uv)
& = & \exp\left (  \frac{\left \| \mu(\zerov) + \zv - \mu(\vv'' - \uv)\Hm) \right \|^2}{-N_0} \right )  P_U(\vv'') \nonumber \\
& = & \exp\left (  \frac{\left \| \mu(\zerov) + \zv - \mu(\cv'' - \xv) \right \|^2}{-N_0} \right )  P_U(\vv''),\,
\end{eqnarray}
where we used the change of variable $\vv' + \uv = \vv''$ and we defined $\cv'' = \vv'' \Hm$.

The error region of (\ref{map-1}) is given by:
\begin{equation}
\Ec(\uv) = \left \{ \left . \zv \in \CC^n : \widehat{\uv} \neq \uv \right | \uv \;\; \mbox{is generated by the source} \right \}
\end{equation}
For the same realization of $\uv$, the error region of (\ref{map-2}) when the all-zero codeword is transmitted is given by:
\begin{equation}
\Ec_0(\uv) = \left \{ \left . \zv \in \CC^n : \widehat{\vv} \neq \zerov \right |  \uv \;\; \mbox{is the discrete channel noise} \right \}.
\end{equation}
By comparing (\ref{zio2}) and (\ref{zio1}) and noticing that the sets of vectors
$\left \{ (\xv' - \xv, \uv') : \uv' \in \GF(q)^K \right \} $ and  $\left \{ (\cv'' - \xv, \vv'') : \vv'' \in \GF(q)^K \right \} $ are identical,  we have that
if $\zv \in \Ec_0(\uv)$ then $w_{-\xv}(\zv) \in \Ec(\uv)$ and, vice versa, if $\zv \in \Ec(\uv)$ then $w_{\xv}(\zv) \in \Ec_0(\uv)$.
Since $w_{\xv}$ is an isometry of $\CC^n$, the congruence of the error regions
$\Ec(\uv)$ and $\Ec_0(\uv)$ is established.
\end{proof}

By noticing that the Gaussian distribution is invariant with respect to isometries, the conditional probability of error for the joint source-channel coding
scheme $\PP(\zv \in \Ec(\uv) | \uv)$ and for the associated channel coding scheme $\PP(\zv \in \Ec_0(\uv) | \uv)$ are identical,
for all realizations of the source vector $\uv$.
Furthermore, we also show in Appendix \ref{sec:app-iso} that a similar equivalence holds for the suboptimal Belief Propagation (BP) decoder \cite{Urb_Book},
in the sense that at every iteration of the decoder, the set of messages generated by the message-passing BP decoder
for the source-channel coding scheme can be mapped into the corresponding set of messages generated by the message-passing BP decoder
for the associated channel coding scheme by a probability-preserving mapping \cite{OZGUN-THESIS}.
It follows that good systematic codes for the two-block composite channel (either under MAP decoding or under BP decoding)
yield immediately good codes (with identical performance) for the source-channel coding problem.
Notice that, with no restriction on decoding complexity and block length,  successful decoding can be achieved with high probability
if $n > K H/C_{\mathfrak{X}}(E_s/N_0)$, which is also the Shannon limit for the two-block composite channel.

Focusing on practical coding design with affordable complexity, the proposed coding optimization strategy consists of choosing a family of
good systematic codes under BP decoding for the two-block composite channel.
Since the source entropy varies from image to image, across the source
components $i$ (DWT subbands) and symbol planes $p$,
it is necessary to choose families of codes spanning a very wide range of coding rates.
Systematic Raptor codes are ideal candidates for this application since they can produce parity
symbols ``on demand'', and cover a continuum of coding rates. In addition, they have excellent performance under BP decoding.
Non-universality of Raptor codes for general noisy channels is well-known (see \cite{EtSho06}), and it is established by the fact that  the {\em stability condition} on the fraction of degree-$2$ output nodes depends on the channel parameter. Following the approach of \cite{EtSho06}, in Appendix \ref{sec:app-stability}, we extended the stability condition to the case of the
two-block composite channel and $q$-ary Raptor codes. It turns out that in this case the stability condition is a function of both the symbol plane entropy $H = H_p^{(i)}$ and the channel capacity
$C = C_{\mathfrak{X}}(E_s/N_0)$. Hence, the Raptor degree distribution must be optimized for each pair of $(H,C)$ values.
We perform this optimization using ``EXIT charts'' and linear programming, extending
\cite{EtSho06} and \cite{Bennatan-nonbinary-journal} to handle the two-block composite channel.

Before entering the details of the EXIT chart analysis and Raptor code optimization,
a final remark on the signal constellation is in order. Since the dead-zone quantizer symbols
are ternary, $q$ must be at least 3. We considered both $3$-PSK and QPSK (with Gray Mapping) constellations.
Although $3$-PSK is more naturally matched to the ternary source alphabet, the QPSK modulation has higher capacity.
Hence, it is not a priori obvious which of the two constellation performs better in our context.
In our experiments we observed that QLIC with $q = 3$, using the $3$-PSK constellation,
did not provide any improvement over the case $q = 4$ with the QPSK constellation. Since QPSK with Gray mapping is standardized in
deep-space communications, and the BP decoder is simplified for powers of 2 field size (see \cite{Bennatan-nonbinary-journal}),
$q = 4$ represents a better and more natural choice. Thus, in the following we only focus on QPSK and codes over $\GF(4)$.
 \subsection{EXIT Chart Analysis for the Two-Block Composite Channel} \label{sec:EXIT}
\begin{figure}[htpb]
\centerline{\includegraphics[width=10cm]{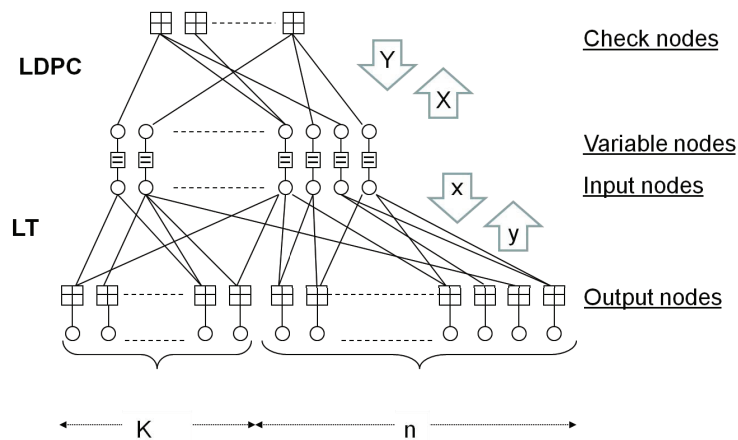}}
\caption{The Tanner Graph of a Raptor Code with LDPC code.} \label{fig:DecRap2}
\end{figure}

We assume that the reader is familiar with Raptor codes, their systematic encoding and
BP iterative decoding, and with the Gaussian approximation EXIT chart analysis technique of BP decoding for
standard binary codes over memoryless binary-input output-symmetric channels (see \cite{EtSho06}).
Here, we focus on the aspects specific to our problem.

A Raptor code is formed by the concatenation of a pre-code, here implemented by a high rate regular LDPC code,
and an ``LT'' code, which is a low-density generator matrix code with a special generator matrix degree distribution
\cite{EtSho06}.  For the Tanner graph of the LT code, we define the input nodes and the output nodes.
For the Tanner graph of the LDPC code, we define the variable nodes and the check nodes
(see Fig. \ref{fig:DecRap2}).
We consider  Raptor codes over $\GF(4)$ with systematic encoding. The first $K$ output symbols of the Tanner graph of Fig.~\ref{fig:DecRap2}
are the systematic symbols, corresponding to the source block $\uv$. The remaining $n$ output nodes are the non-systematic (parity) symbols,
corresponding to the codeword $\xv$. Thanks to the equivalence of Theorem \ref{map-isomorphism} and to the analogous equivalence for BP decoding
\cite{Ozgun-Maria-JSCC-08}, \cite{OZGUN-THESIS}, we consider the transmission of the Raptor codeword over the two-block composite channel
where the first block of $K$ symbols go through the additive noise over $\GF(4)$ with noise identically distributed as the source vector $\uv$, and the
second block of $n$ symbols is mapped onto QPSK by Gray mapping and is sent through the AWGN channel (\ref{awgn}).
Hence, the rest of this section is dedicated to the Raptor code ensemble optimization (namely, the optimization of its degree distribution) for
the associated two-block composite channel.

For codes over $\GF(4)$ we use the Gaussian approximation approach proposed in \cite{Bennatan-nonbinary-journal}.
In particular, the conditional distribution of each message $\Lm$ in Log-Likelihood Ratio (LLR) domain~\footnote{The BP messages for $q$-ary codes can be either represented as probability vectors  (of length $q$) or
as LLR  vectors of length $q-1$. If $\mv$ is a message in the probability domain,
the corresponding message in the LLR domain, denoted by $\Lm$, has elements  ${L}_i = \log(m_0/m_i)$ for $i=0,\ldots, q-1$.} is assumed to be Gaussian $\Lm \sim \Nc(\upsilon \onev,\Sigma_\upsilon)$,
where $\left[\Sigma_{\upsilon}\right]_{i,j} = 2\upsilon$ for $i=j$ and $\left[\Sigma_{\upsilon}\right]_{i,j} = \upsilon$ for $i\neq j$.
It can be noticed that the conditional distribution depends only on a single parameter $\upsilon$ thanks to symmetry and permutation invariance assumption of the messages as defined in \cite{Bennatan-nonbinary-journal}.
Letting $V$ the code variable corresponding to the edge message $\Lm$, we define the mutual information function
$  J(\upsilon)\eqdef   I(V;\mathbf{L}) = 1 - \EE \left [ \log_4\left (1 + \sum_{i = 1}^3e^{-L_i} \right ) \right ] . $
We use base-4 logarithm for mutual information calculations, hence in these sections $H$ and $C$ are in units of two bits
per source symbol or per channel symbol, respectively.

The EXIT chart is the mapping function of a multidimensional dynamic system that describes the evolution
of the mutual information between the Tanner graph variables and the corresponding messages passed along the
Tanner graph edges by the BP decoder. The stationary points of such dynamic system are given as the solutions of a set of
EXIT chart fixed-point equations, given in terms of the following {\em state variables}:\\
- $\x$ denotes the average mutual information between a randomly chosen input node symbol and
the corresponding message sent downward to an adjacent edge (from input to output nodes). See Fig. \ref{fig:DecRap2}.\\
- $\y$ denotes the average mutual information between a randomly chosen input node symbol and
the corresponding message received upward from an adjacent edge (from output to input nodes).\\
- $\X$ denotes the average mutual information between a randomly chosen
variable node symbol and the corresponding message sent upward to an adjacent edge (from variable to check nodes).\\
- $\Y$ denotes the average mutual information between a randomly chosen
variable node symbol and the corresponding message received downward from an adjacent edge (from check to variable nodes).\\
The degree distributions for the Tanner graph in Fig.~\ref{fig:DecRap2} are defined as follows:\\
- For the LDPC code, we let $\lambda(x) = \sum_i \lambda_i x^{i-1}$ and $\rho(x) = \sum_j \rho_j x^{j-1}$
denote the generating functions of the edge-centric left and right degree distributions, and we let
\[ \Lambda(x) = \sum_i \Lambda_i x^i = \left(\int_0^x \lambda(u) du\right)/\left(\int_0^1\lambda(u) du\right), \]
denote the node-centric left degree distribution.\\
- For the LT code, we let $\iota(x) = \sum_i \iota_i x^{i-1}$ denote the edge-centric degree distribution
of the input nodes, and we let $\omega(x) = \sum_j \omega_j x^{j-1}$ denote the edge-centric degree distribution of the ``output nodes''.
The node-centric degree distribution of the output nodes is given by
\[ \Omega(x) = \sum_i \Omega_j x^j = \left(\int_0^x \omega(u) du\right)/\left(\int_0^1 \omega(u) du\right). \]
- For the concatenation of the LT code with the LDPC code we also have the node-centric degree distribution of the LT input nodes.
This is given by
\[ \gimel(x) = \sum_i \gimel_i x^i = \left(\int_0^x \iota(u) du\right)/\left(\int_0^1 \iota(u) du\right). \]
Note that for large number of nodes we have the following approximation for $\gimel(x)\sim e^{\alpha(x-1)}=\sum_n\frac{\alpha^n e^{-\alpha}}{n!}x^n$ where $\alpha = \sum_i \gimel_i i$ is the average node degree for the input nodes~\cite{EtSho06}.
Hence $\iota(x)$ is approximated by the following coefficients
\begin{equation}\label{eq:iota_i}
\iota_i = \alpha^{i-1}e^{-\alpha}/(i-1)!.
\end{equation}
The capacities of the first and second components of the two-block composite channel are $1 - H$ and $C$, respectively.
A random edge $(\sf o, \sf v)$ is connected with probability $\gamma = K/(K+n)$ to the first block and
with probability $1 - \gamma$ to the second block.  As a consequence, it is just a matter of a simple exercise to obtain
the following EXIT equations for the LT code component (Detailed derivations for binary Raptor codes can be found in our previous work \cite{Ozgun-Maria-JSCC-08}):
\begin{eqnarray} \label{ldgm-down}
\x &  = &  \sum_k \sum_i \Lambda_k  \iota_i J( (i-1) J^{-1}(\y) + k J^{-1}(\Y)),\\
 \label{ldgm-up}
\y & = & 1 - \sum_j \omega_j {\Big\{}\gamma J( (j-1) J^{-1}(1 - \x) + J^{-1}(H))+(1 - \gamma) J( (j-1) J^{-1}(1 - \x) + J^{-1}(1 - C))\Big{\}}. \nonumber \\
& &
\end{eqnarray}
Also, notice that $\gamma = r_{\rm lt} r_{\rm ldpc}$,  where $r_{\rm lt} = \frac{1}{\alpha\sum_j\omega_j/j}$ and
$r_{\rm ldpc} = 1 - \frac{\sum_i i\lambda_i}{\sum_j j \rho_j}$ are the coding rates of the LT code and of the LDPC code,
respectively.

The EXIT equations for the LDPC component are well-known and are given by:
\begin{eqnarray} \label{ldpc-up}
\X & =  & \sum_k \sum_i \lambda_k  \gimel_i J( (k-1) J^{-1}(\Y) + i J^{-1}(\y)),\\
\label{ldpc-down}
\Y &=& 1 - \sum_\ell \rho_\ell J( (\ell - 1) J^{-1}(1 - \X)).
\end{eqnarray}
Eventually, (\ref{ldgm-down}), (\ref{ldgm-up}), (\ref{ldpc-up}), and (\ref{ldpc-down}) form the system of fixed-point
equations describing the stationary points the EXIT chart of the concatenated LT -- LDPC graph,
with parameters $H, C$ and $\gamma$, and the degree sequences
$\omega, \iota, \rho$ and $\lambda$.

The error probability of the output nodes corresponding to the first block of $K$ output nodes,
sent through the discrete additive noise component of the two-block channel,
is identical to  the error probability of the source symbols in the source-channel equivalent problem.
Therefore, the key quantity of interest for the performance of the JSCC scheme is the error probability of such output nodes. This is
can be obtained, within the assumptions of EXIT chart approximation, as follows.
The mean of the LLR of such an output node of degree $j$ is given by
$\upsilon_j = J^{-1} \left (1 - J\left ( j J^{-1}(1 - \x) \right ) \right )  + J^{-1}(1 - H)$.
By the channel symmetry and the code linearity, the EXIT chart is derived under all-zero codeword assumption.
Hence, decoding is successful if the LLR vector has positive components.
For an output node of degree $j$, this results in the symbol error rate (SER) $1-\PP(\Lm\geq \zerov)$, with $\Lm \sim \Nc(\upsilon_j \onev,\Sigma_{\upsilon_j})$.
By averaging over the degree distribution, the desired average SER is given by
\begin{equation} \label{BER}
P_e = \sum_j \Omega_j \left[1-Q^3\left( -\frac{\sqrt{\upsilon_j}}{2} \right)\right].
\end{equation}

\subsection{LT Degree Distribution Optimization}\label{sec:degree_opt}

For simplicity, we fix the LDPC code to be a regular  $(2,100)$ code  ($r_{\rm ldpc} = 0.98$).
For this LDPC code, we find the mutual information threshold $\Yc_0(\alpha)$ (using (\ref{ldpc-up}) and (\ref{ldpc-down}))
such that for $\y \geq \Yc_0(\alpha)$ the LDPC EXIT converges to  $\Y = 1$, with stand-alone iterations.
The value of $\Yc_0(\alpha)$ depends on the LT input degree distribution $\iota(x)$, which in turns depends on
$\alpha$ via (\ref{eq:iota_i}). The function $\Yc_0(\alpha)$ is monotonically decreasing. Therefore, higher values of $\alpha$ yield
less restrictive requirements for the mutual information that the LT code must attain in order to allow the LDPC code to converge to
$\Y = 1$ (vanishing error probability). On the other hand, larger values of $\alpha$ yield smaller
LT coding rate $r_{\rm lt}$, and therefore are more conservative with respect to the system bandwidth efficiency.


Next, we use (\ref{ldgm-down}) and (\ref{ldgm-up}) to eliminate $\x$ and write $\y$ recursively.
The fixed-point equation for $\y$ depends on the input $\Y$ coming from the LDPC graph. In order to obtain a tractable problem,
we decouple the system of equations (\ref{ldpc-up}-\ref{ldpc-down}) and (\ref{ldgm-down}-\ref{ldgm-up})
the target mutual information $\Yc_0(\alpha)$ and disregarding the feedback from LDPC to LT in the BP decoder (i.e., letting $\Y = 0$ in (\ref{ldgm-down}).
This is equivalent to running BP with the following schedule: first iterate the LT code till convergence, and then iterate the LDPC code till convergence.
The resulting recursion mapping function $f_j^{H,C,\gamma,\alpha}(\y)$  for a degree-$j$ output node is given by
\begin{eqnarray}\label{eq:recursion-y}
f_j^{H,C,\gamma,\alpha}(\y)&\eqdef&{\Big\{}\gamma J\left( (j-1) J^{-1}(1 -  \sum_i  \iota_i J( (i-1) J^{-1}(\y))) + J^{-1}(H)\right) \nonumber\\
&&+(1 - \gamma) J\left( (j-1) J^{-1}(1 - \sum_i  \iota_i J( (i-1) J^{-1}(\y) )) + J^{-1}(1 - C)\right)\Big{\}}.
\end{eqnarray}
We conclude that the LT EXIT recursion converges to the target $\Yc_0(\alpha)$ if
\begin{equation}\label{eq:LT-conv-rec}
\y < 1 - \sum_j \omega_j f_j^{H,C,\gamma,\alpha}(\y),\;\;\forall \y\in\left[0,\Yc_0(\alpha)\right].
\end{equation}
In order to ensure this condition,  we sample the interval $[0,\Yc_0(\alpha)]$ on a sufficiently fine grid of
points $\{\y_i\}$, and obtain a set of linear constraints for the variables $\{\omega_j\}$.
The code ensemble optimization consists of maximizing $r_{\rm lt}$ for given $H,C$ pair, subject to the condition that the BP decoder converges to vanishing
error probability. The optimization variables are $\{\omega_j\}$ and $\alpha$. Since the LDPC code is fixed, $\gamma$ is a function of
$\alpha, \{\omega_j\}$. In order to linearize the constraints in $\{\omega_j\}$ we replace
$\gamma$ in (\ref{eq:LT-conv-rec}) with its ideal value $C/(C+H)$, arguing that good codes must have $\gamma \approx C/(C+H)$.
This yields the optimization problem:
\begin{eqnarray} \label{eq:linprog}
\mbox{min}_\alpha &\mbox{min}_{\{\omega_j\}} &  \alpha\sum_j\frac{\omega_j}{j}  \nonumber\\
&\mbox{s. t.}  & \sum_j \omega_j = 1,\;\omega_j\geq 0,\nonumber \\
&& \y_i<1-\sum_j \omega_j f_j^{H,C,\small{\frac{C}{C+H}},\alpha}(\y_i), \;\;\; \forall \; \y_i\in[0,\Yc_0(\alpha)].
\end{eqnarray}
For fixed $\alpha$,  (\ref{eq:linprog}) is a linear program with respect to $\{\omega_j\}$. Hence, we can run an educated line search with respect to
the scalar variable $\alpha$ and, for each tentative $\alpha$, easily optimize over  $\{\omega_j\}$.

Let $\Psi(\upsilon)\eqdef \mathbb{E}\left[\frac{1-e^{-L_1}+e^{-L_2}-e^{-L_3}}{1+e^{-L_1}+e^{-L_2}+e^{-L_3}}\right]$, then the stability condition obtained in Appendix \ref{sec:app-stability} reads \cite{OZGUN-THESIS}:
\begin{equation}\label{eq:stability}
\Omega_2 \geq \frac{\gamma/r_{\rm ldpc}}{2\left(\gamma \Psi\left(J^{-1}(1-H)\right)+(1-\gamma)\Psi\left(J^{-1}(C)\right)\right)}\;.
\end{equation}
We accept the solution of the optimization
if (\ref{eq:stability}) is satisfied. Otherwise, the optimization is re-run with a more conservative value of $\alpha$.

\section{Results}\label{sec:Results}

We present the performance of the QLIC scheme and compare it with the
baseline (state-of-the art) system used by Jet Propulsion Laboratories (JPL) for the Mars Exploration Rover (MER) Mission.
For the purpose of this comparison, we briefly present the current baseline system.
The scheme is based on a separated source and channel coding approach, concatenating
an image coding scheme called ICER \cite{kiely-icer-report} with standard codes for deep-space communications
\cite{jpl-codes-ieee}, \cite{ORange}.

ICER is a successive refinement, wavelet-based image compressor based on the same principles
of JPEG2000, including image segmentation, DWT, quantization, and entropy coding of the blocks of quantization indices
using interleaved entropy coding and an adaptive probability estimator based on context models \cite{kiely-icer-report}.
These components differ from  their JPEG2000 counterparts in order to handle specific needs of scientific images for
deep-space exploration.  ICER makes use of a reversible integer-valued DWT \cite{pearlman}
so that, if all the subbands data are fully transmitted, lossless reconstruction can be obtained.
Since subband coefficients are integer values, the dead-zone quantizer is also modified to work for integer
values as described in \cite{kiely-icer-report}. The quantization precision for each subband and the selection of which subbands should be transmitted
in order to minimize the total number of bits subject to a given target reconstruction Peak SNR~\footnote{The reconstruction PSNR is defined as
$\PSNR = 10\log_{10}\frac{2^{i}-1}{D}$ where $D$ is the WMSE distortion and where $i=12$, since for MER mission each pixel
is a $12$-bit value in the original image.} (PSNR) are established dynamically, based on the actual image to be encoded,
according to the relative importance of each subband.
The resulting priority-ordered bit planes are encoded one by one, until the target PSNR (or total bit budget) is reached.
ICER and JPEG2000 (using either lossless-5/3 integer DWT or lossy-9/7 DWT) provide
similar pure image compression performances, i.e., when used on noiseless channels \cite{kiely-icer-report}.
In \cite{JPL-report} we found that the pure compression performance of the proposed QLIC scheme is also almost
identical to ICER. This provides a good sanity check for QLIC, which is not inferior to the state-of-the art as far as
pure image compression is concerned.

In order to increase robustness against channel errors, ICER partitions the image into segments.
A segment ``loosely'' corresponds to a rectangular region of the image (although in practice a more sophisticated adaptive segmentation scheme is used).
Each image segment is compressed independently. In this way, the error propagation introduced by possible residual post-decoding channel errors is limited to within a segment.  The encoded bits corresponding to the segments are concatenated and divided into
fixed-length frames, that are separately channel-encoded at channel coding rate $R_c$. This takes on values in a finite set of
possible coding rates supported by the family of deep-space channel coding schemes.
The channel coding rate $R_c$ is chosen according to the channel SNR.
The presence of residual errors is detected with probability close to 1 using standard error detection techniques, and
the frames with residual post-decoding errors are erased. Frame erasure is the main cause of data loss in the baseline
JPL scheme \cite{kiely-icer-report}.

Since ICER is a progressive image compressor, all successfully decoded frames of a segment before the first erased
frame can be used for source reconstruction. The reconstruction quality of a segment depends on the position of the first
frame erasure (of course, the highest quality is obtained if no erasure occurs). No unequal error protection is used for the sequence of
successive frames forming a segment. Therefore, all frames have the same erasure probability \cite{kiely-icer-report}.
As a consequence, segments may be reconstructed at very different quality level, depending on the presence and position
of frame erasures. If a segment achieves too poor reconstruction quality, the retransmission of the whole segment is requested.
In MER, retransmissions are possible with a delay roughly equal to the round trip time between Earth and Mars which is between 7-30 minutes \cite{round-trip}.
In addition, retransmissions  require storing the images on the deep-space probe, for a long time and a feedback channel from Earth to Mars
for retransmission requests.

In our comparison, we consider just the spectral efficiency of the ``active transmission'' phase, i.e., as defined by the parameter $b$.
This is the ratio between channel uses (including retransmissions) and source samples.
The comparisons reported here do not take into account the long idle times and the enormous delay incurred
by retransmissions, because the ``cost'' of these system aspects is difficult to quantify from a communication theoretic viewpoint.
However, we hasten to say that the proposed JSCC scheme {\em does not} require retransmissions unless the channel SNR dramatically changes
with respect to the nominal value assumed at the transmitter. Hence, although the spectral efficiency performance is slightly inferior
to the baseline system, this built-in robustness able to avoid retransmissions is a very attractive feature in terms of delay and system simplification.
\begin{figure}[htpb]
\centerline{\includegraphics[width=8cm,height=6cm]{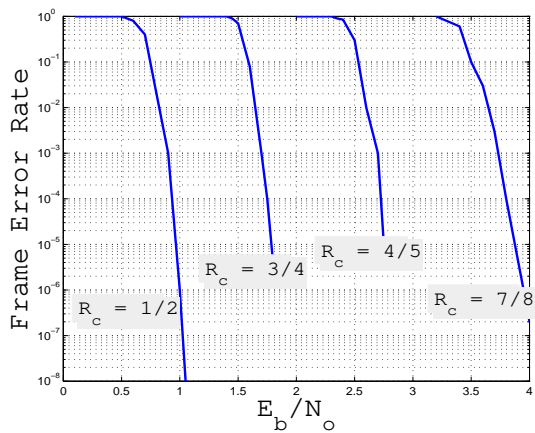}}
\caption{FER vs $E_b/N_o$ curves for block length $16$Kb.}
\label{fig:FER-curves}
\vspace*{-0.3 cm}
\end{figure}
First, we considered a scenario where the target PSNR is fixed.
For a given set of test images~\footnote{Provided by JPL-MER Mission Group.},
we compare the two schemes in terms of $b$ versus $E_s/N_0$, for the same target PSNR.
The Frame Erasure Rate (FER) of the baseline system
is a function of the channel SNR and of the channel coding rate $R_c$ used.
The number of transmissions necessary for the successful reconstruction of a segment
is a geometric random variable with success probability that depends on the FER, the number
of frames $F$ in a segment, and on the target PSNR.
As mentioned before, the reconstruction quality of a segment depends on the position of the
first erased frame in the segment. Upper and lower bounds to the success probability developed in \cite{JPL-report} show that,
for the typically very high target PSNR required by deep-space scientific imaging, the success probability is tightly approximated by
$(1 - {\sf FER})^{F}$
(i.e., a segment is retransmitted whenever a frame is in error, irrespectively of its position).
For a given channel SNR, the baseline scheme chooses the deep-space channel code with maximum rate $R_c$, subject
to the condition that the FER must be smaller than a target threshold (a typical target is $10^{-6}$).
The target FER is fixed in order to achieve a desired, and typically very small, retransmission probability.
For a well matched SNR and rate pair, the FER is effectively very small and the expected number of re-transmission
is insignificant. In this case, $b$ is very close to the ``one-shot'' transmission value,
i.e. $B/(2R_c)$,  where $B$ is the number of ICER-encoded bits per pixel at the given target PSNR
and the factor 2 comes from the fact that QPSK transmits 2 coded bits per channel use.

For each fixed $R_c$, the corresponding $b$ vs. $E_s/N_0$ curve has a very pronounced ``L'' shape,
due to the sharp waterfall of the FER (see Fig. \ref{fig:FER-curves}).
Hence, the SNR axis can be split into intervals, where each interval corresponds to the range of $E_s/N_0$ values
for which a given coding rate ``dominates'', i.e. yields the best efficiency (including retransmissions). If $E_s/N_0$ is known in advance,
and the cost of retransmissions is neglected, for each SNR falling in a given interval, the corresponding coding rate is selected.
For the example considered in this paper, the target PSNR is $49$ dB and the MER1 image ($1024\times 1024$ BW uncoded 12-bit per pixel)
is used. Fig. \ref{fig:Jan21}-A compares the resulting $b$ vs $E_s/N_0$ performance of JSCC and of the baseline scheme.
The following comments are in order:\\
- The ($\ast$)-curve corresponds to considering ideal capacity achieving codes for each symbol plane
in the JSCC scheme. This represents the best possible performance for the DWT and quantization scheme used in the proposed system.
This curve is also a very good approximation of the performance of a separated scheme based on QLIC or ICER for pure compression,
concatenated with an ideal capacity achieving channel code, since the pure source compression performance of ICER and QLIC is essentially indistinguishable
\cite{JPL-report}.\\
- The performance of the actual baseline scheme is shown as a superposition of four L-shaped curves, each of which corresponding to one
of the codes whose FER performance is shown in Fig. \ref{fig:FER-curves}, as explained before.
The steep increase of $b$ for small degradation of $E_s/N_0$ beyond the ``knee'' point of each L-shaped curve indicates that
if the channel quality degrades slightly below the threshold at which each channel code yields small FER, then the number of retransmissions
per segment increases dramatically. If such channel quality degradation occurs (e.g., atmospheric propagation phenomena, rain conditions, antenna alignment fluctuations),
then the conventional system folds back onto a more conservative channel coding rate, and its performance moves on
the L-shaped curve to the left.  The channel coding rate value $R_c$  corresponding to each curve
is also shown in Fig. \ref{fig:Jan21}-A.\\
- Before designing degree distributions specifically for the Raptor codes over $\GF(4)$ as described in Sec.~\ref{sec:EXIT}, we first used the degree distribution
\begin{eqnarray}\label{eq:classic_LT}
\Omega(x) & = & 0.008x + 0.494x^2 + 0.166x^3 + 0073x^4 + 0.083x^5 \nonumber \\
& & + 0.056x^8 + 0.037x^9 + 0.056x^{19} + 0.025x^{65} + 0.003x^{66},
\end{eqnarray}
given in \cite{shokr2006IT} for binary Raptor codes. The EXIT chart infinite-length performance and finite length simulations for this
non-optimized non-binary case are shown in Fig.~\ref{fig:Jan21}-A as the
($-.$)-curve and the ($\square$)-curve, respectively.\\
- In Sec.~\ref{sec:code-design} we noted that the proposed method for the Raptor code degree distribution optimization in (\ref{eq:linprog})
is a linear program when the parameter $\alpha$ and the LDPC code are fixed. As seen in Sec.~\ref{sec:degree_opt},
$r_{\sf lt}$ is a function of $\alpha$ and $\omega(\cdot)$.
Hence, the aim is to maximize $r_{\sf lt}$ by optimizing $\omega(\cdot)$ for fixed $\alpha$, at each value of $E_s/N_0$. The value of
$\alpha$ for given $E_s/N_0$ is obtained using the non-optimized code simulations as follows:
for the non-optimized (($-.$) and ($\square$)) cases the distribution $\omega(\cdot)$ is given by (\ref{eq:classic_LT})
and $r_{\sf lt}$ is provided by the simulation at each $E_s/N_0$ point.  Then, we obtain the corresponding $\alpha$
by using $\omega(\cdot)$ and $r_{\sf lt}$. Finally, for this fixed pair of $\alpha$ and $E_s/N_0$, we use the linear program in order to
optimize the LT degree distribution.
Although an exhaustive search over the feasible range of $\alpha$ may yield further improvements, the above simple
method already provides a noticeable performance enhancement both in terms of the EXIT chart infinite-length performance
($\circ$) and in terms of the finite-length simulation ($\diamond$), with respect to the corresponding
non-optimized curves ($-.$) and ($\square$).\\
Next, we focus on a particular channel SNR value (in particular, we choose $E_s/N_0 = 3$ dB),
and provide a zoomed version of Fig. \ref{fig:Jan21}-A around this value in Fig. \ref{fig:Jan21}-B.
For this SNR, the channel code with rate $R_c = 3/4$ yields the best performance for the conventional system.
Now we consider the case of a mismatch between the actual and the nominal channel quality, i.e., we assume that the transmitter chooses the optimal scheme
(baseline or proposed JSCC) for the nominal $E_s/N_0 = 3$ dB, but the actual value of $E_s/N_0$ is less than 3 dB.
In this case, the efficiency $b$ of the baseline system significantly decreases
due to the retransmissions. At a certain point, as the channel conditions worsens, the baseline system switches to the next lower channel coding rate
$R_c = 1/2$. This happens at $E_s/N_0 \approx 2.8$ dB. The proposed  JSCC scheme has a better built-in robustness to handle
mismatched channel conditions, thanks to the QLIC linear map. We observe that the JSCC scheme optimized for  $E_s/N_0 = 3$ dB and
with no retransmission yields constant $b$ and a slight degradation of the reconstruction PSNR over the range of channel SNR.
Due to mismatched channel conditions, there will be some residual error in the symbol planes.
However, as seen in Fig. \ref{fig:MER1}-b,-c,-d,-e, the perceptive quality of the reconstructed image (and the reconstruction PSNR)
gracefully degrades and the perceived image quality is acceptable over a wide range of channel SNRs, since there are no artificial ``block effects''
due to segment losses, even though the channel SNR is as far as $1$ dB less than  its nominal value of 3 dB.
Similar behaviors have been observed by extensive experimentation, not reported here for the sake of brevity and space constraints.
\begin{figure*}[htpb]
\centerline{(A)\includegraphics[width=7.5cm,height=6cm]{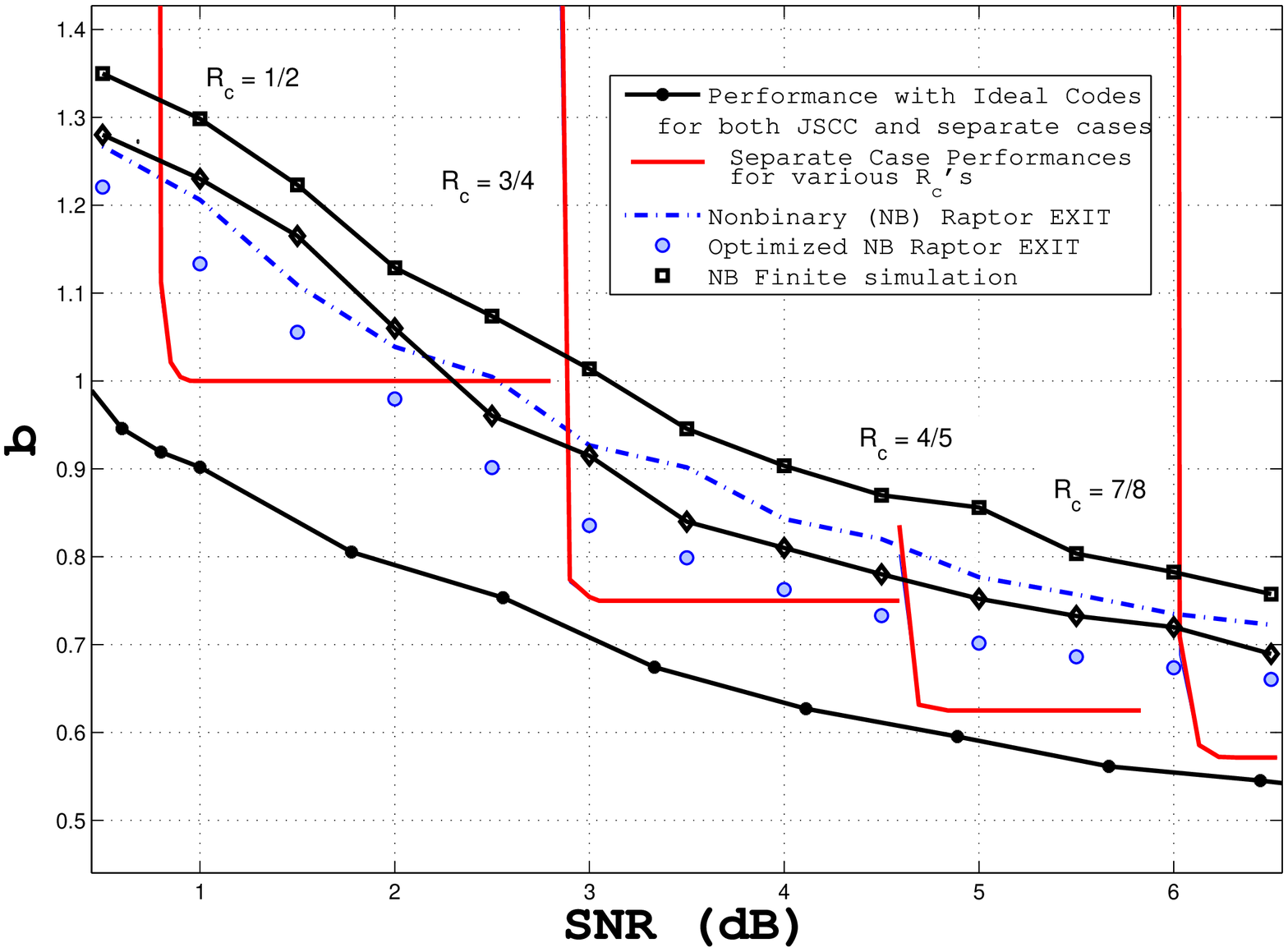}
(B)\includegraphics[width=7.5cm,height=6cm]{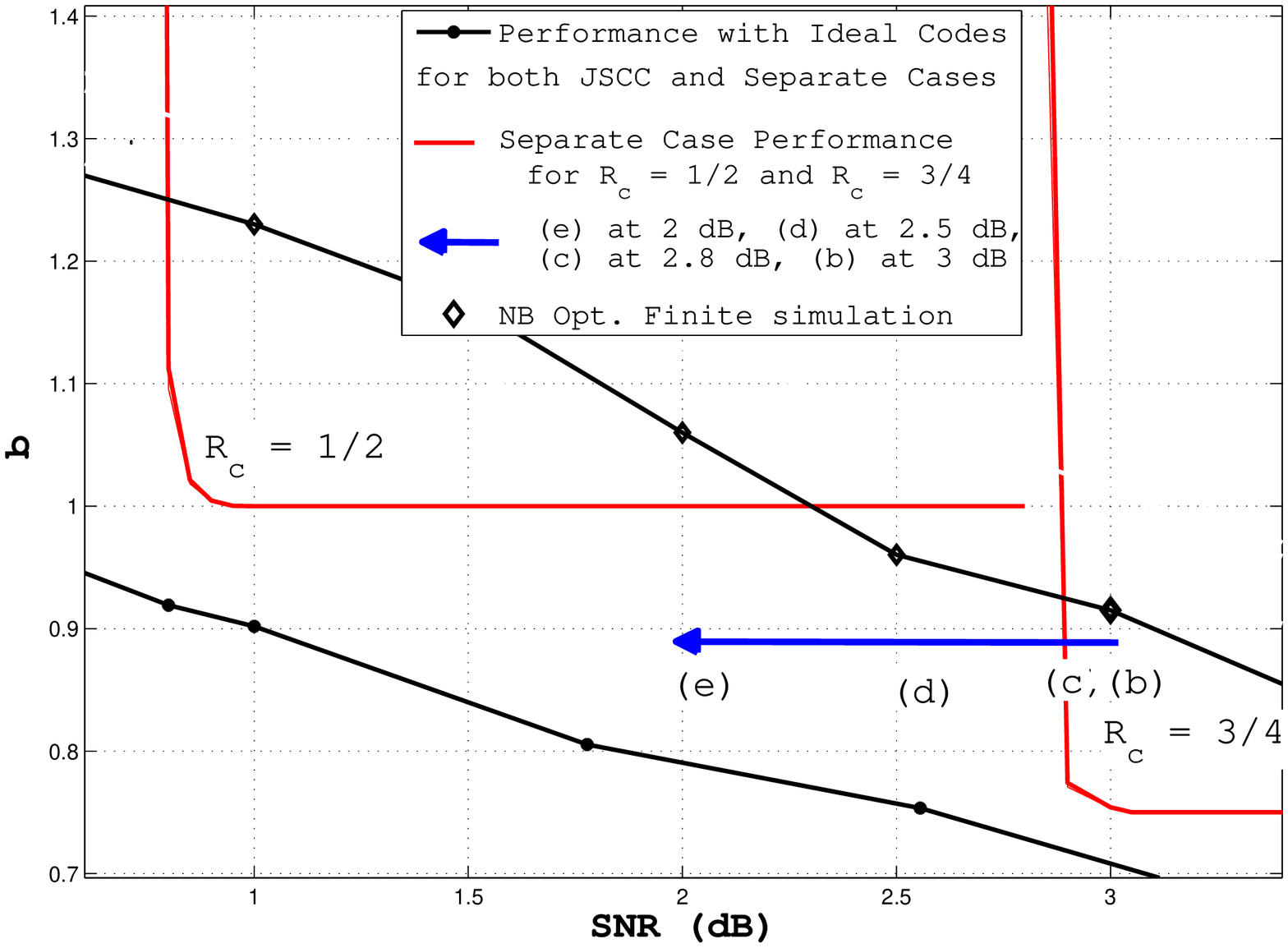}}
\caption{(A): $b$ vs $E_s/N_0$ trade-off curves for various schemes. (B): Focusing on $E_s/N_0 = 3$ dB point. The image reconstructions at various mismatched SNR
values, indicated by (b), (c), (d), (e) in (B), are shown in Fig.~\ref{fig:MER1}.}
\label{fig:Jan21}
\vspace*{-0.3 cm}
\end{figure*}

\begin{figure}[htpb]
\centerline{(a)\includegraphics[width=5.5cm,height=5.5cm]{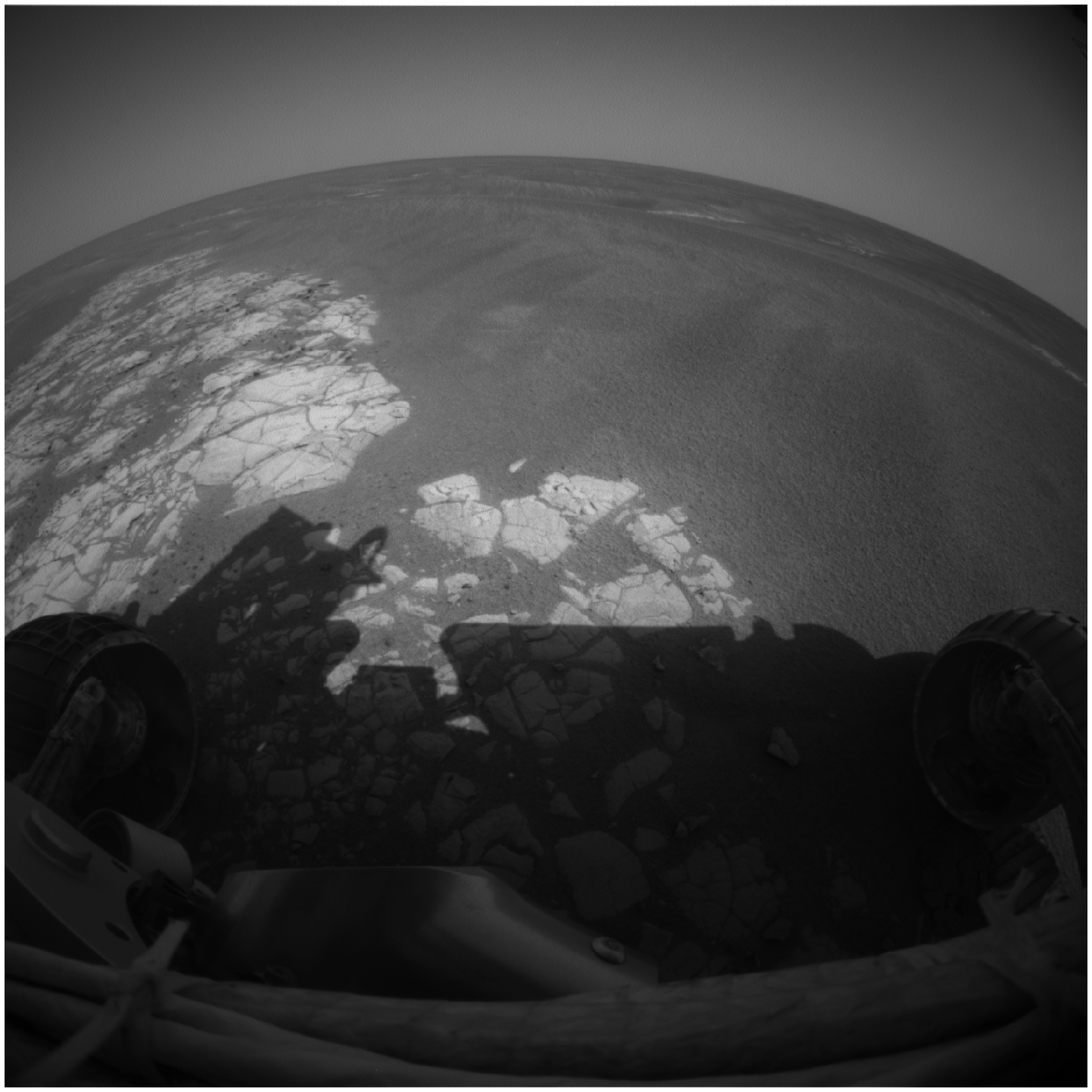}
(b)\includegraphics[width=5.5cm,height=5.5cm]{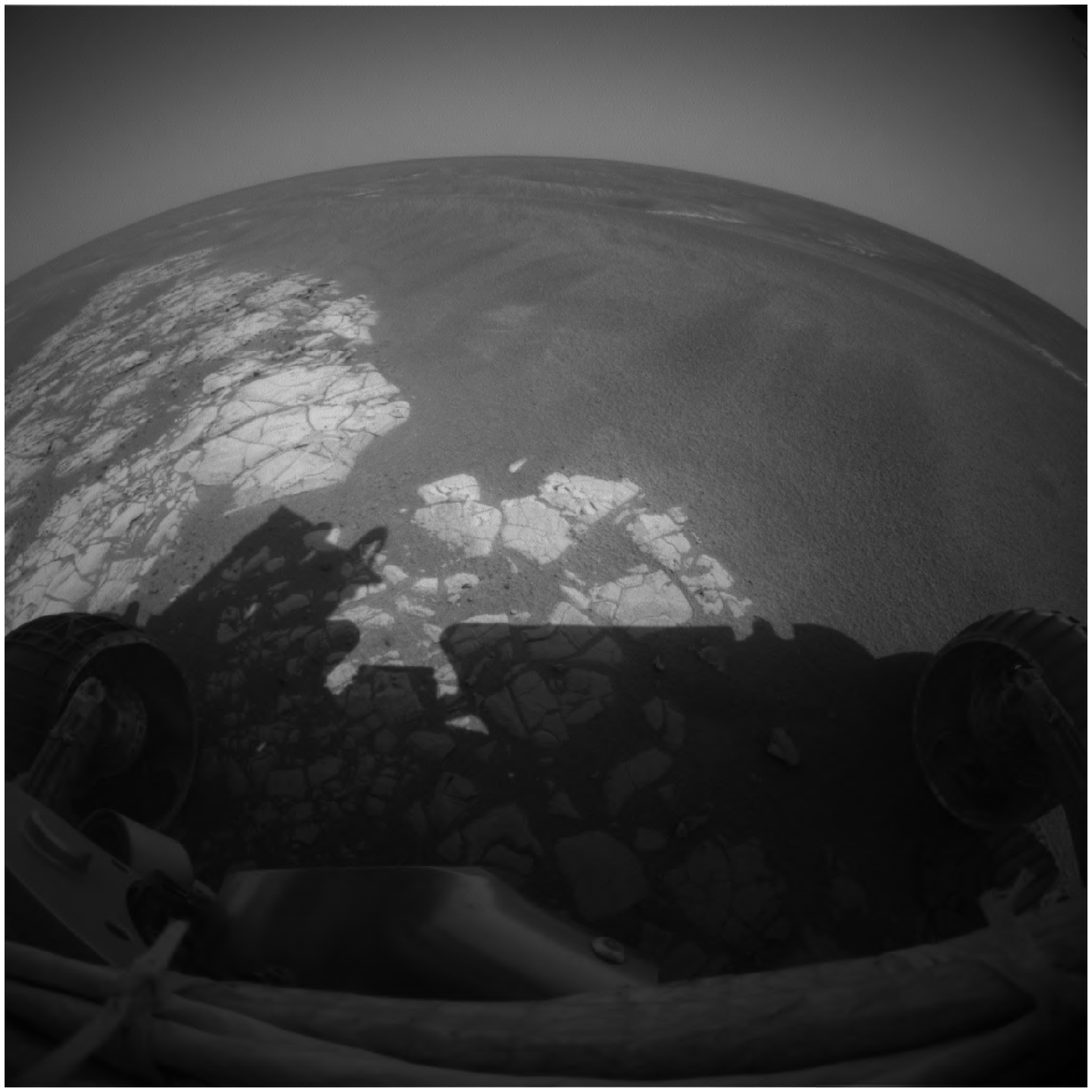}
(c)\includegraphics[width=5.5cm,height=5.5cm]{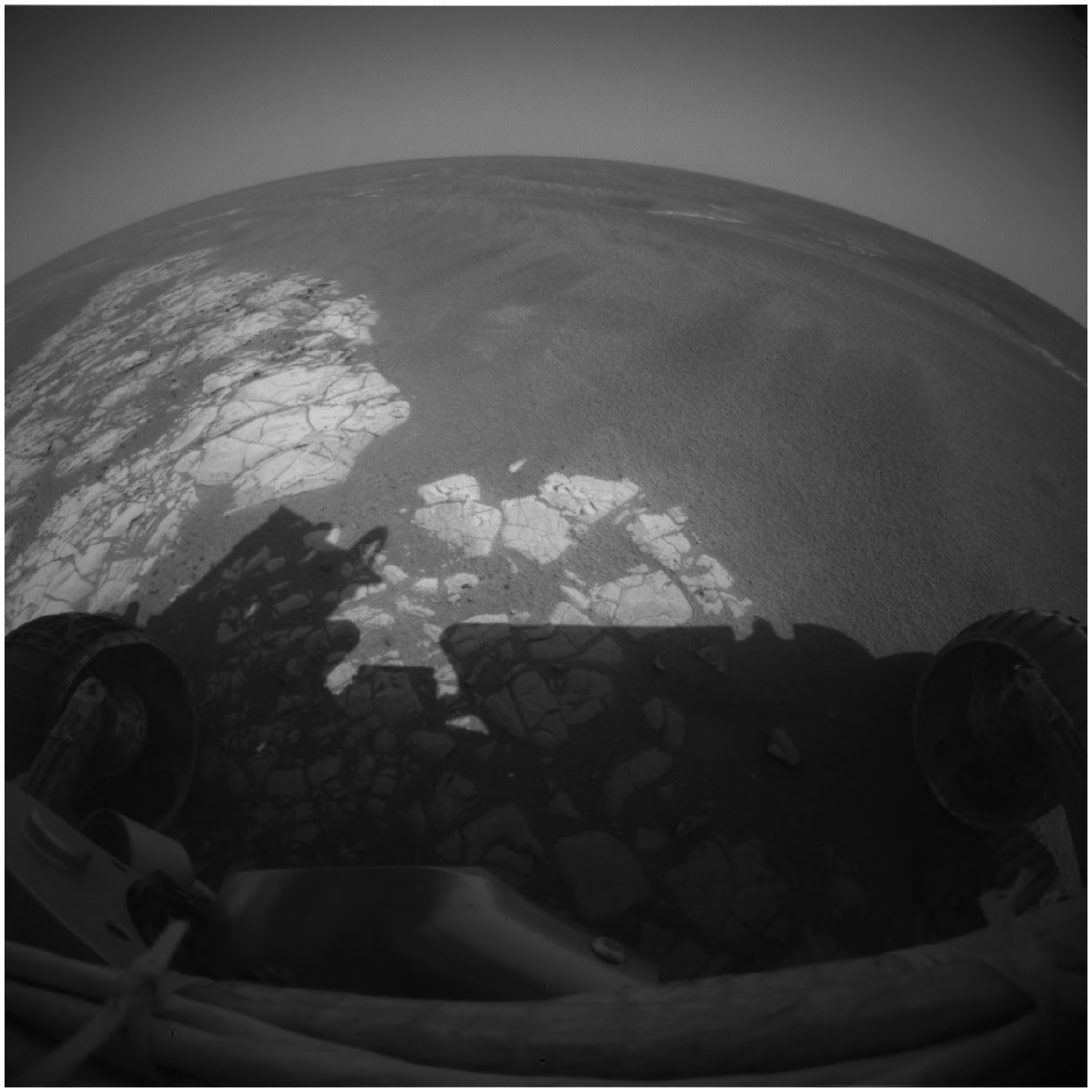}}
\centerline{(d)\includegraphics[width=5.5cm,height=5.5cm]{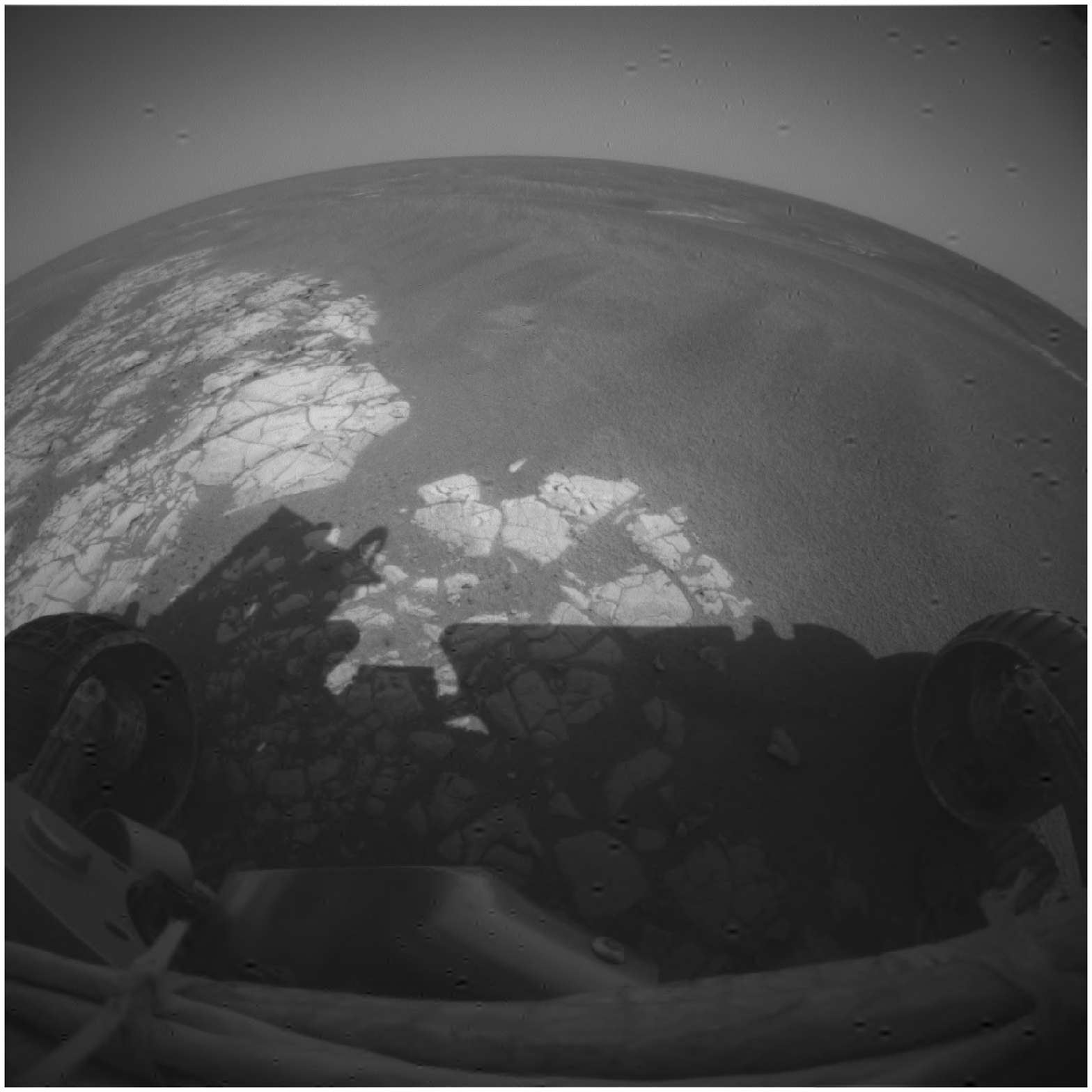}
(e)\includegraphics[width=5.5cm,height=5.5cm]{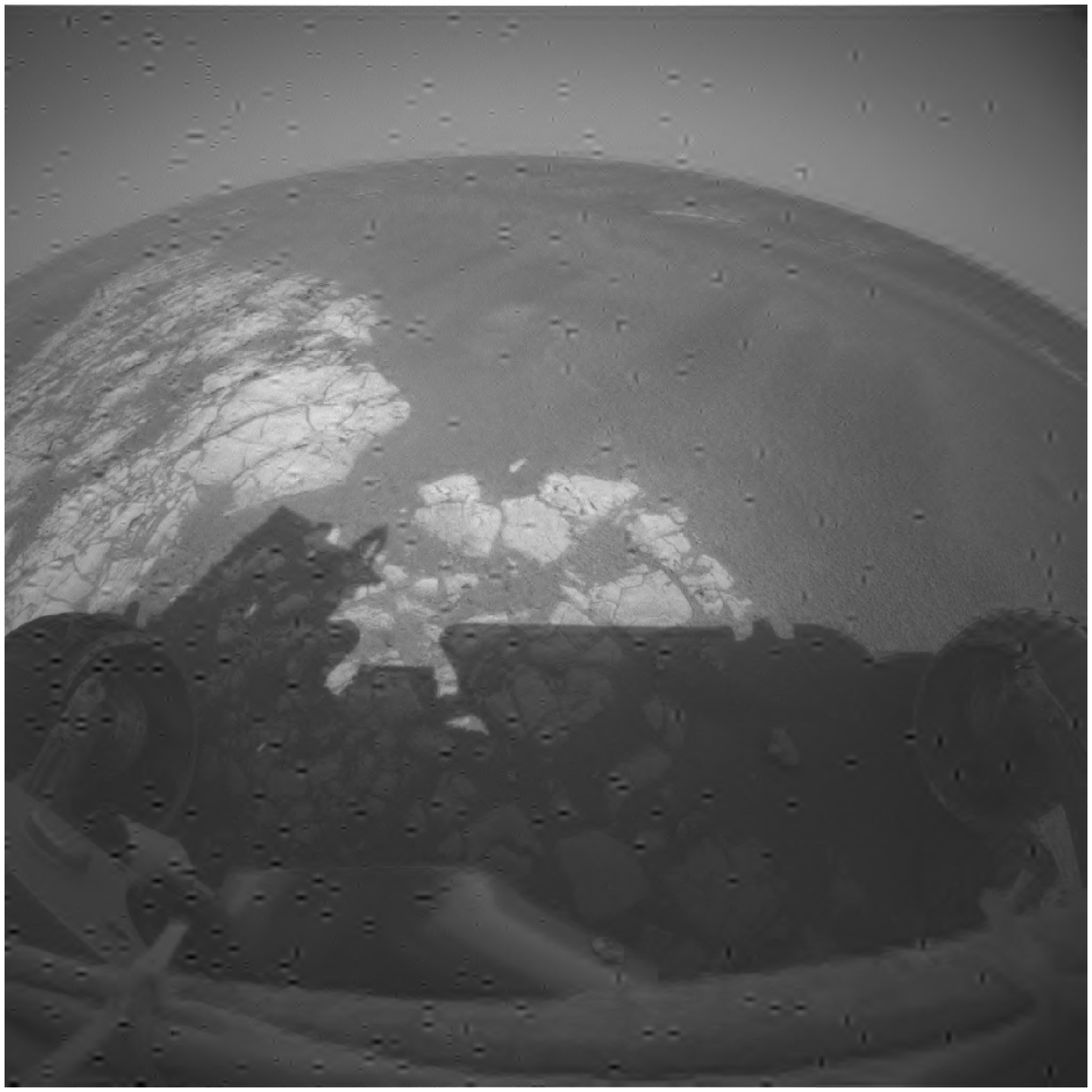}}
\caption{ Original MER1 image (a), and different image reconstructions at $E_s/N_0 = 3$ dB, and PSNR $ = 49$ dB (b);
at $E_s/N_0 = 2.8$ dB, and PSNR $ = 48.19$ dB (c);
at $E_s/N_0 = 2.5$ dB, and PSNR $= 45.63$ dB (d); at $E_s/N_0 = 2$ dB and PSNR $ = 38.60$ dB (e). Notice that visible artifacts
due to residual channel errors can be seen only in figure (e).}
\label{fig:MER1}
\vspace*{-0.3 cm}
\end{figure}

\section{Conclusions}

We proposed a new coding scheme for digital image transmission over a discrete-time AWGN channel.
The scheme is based on the concatenation of a standard DWT, decomposing the image into blocks of subband coefficients,
and an embedded dead-zone quantizer that produces sequences of ternary quantization indices for the successive refinement ``planes'' of
each subband. Then, the redundant symbol planes are mapped linearly into channel codewords, which are modulated
into constellation symbols and sent over the discrete-time AWGN channel.
We showed that if the quantization indices symbol alphabet additive group structure is matched to the
signal constellation generating group structure, the modulation mapping is an isometric labeling, and the source-channel encoder is linear,
then the source-channel coding problem is equivalent to a channel coding problem over a composite two-block channel, where the first
block corresponds to the transmission over a discrete additive noise channel with noise statistics identical to the source
statistics of the original source-channel coding problem, and the second block is the AWGN channel with the isometric labeling included as part of the channel. This equivalence holds for both the optimal MAP decoder and the suboptimal, low-complexity, BP decoder.
This allows us to optimize the source-channel coding ensemble as if it was a channel coding ensemble for the equivalent channel.
In particular, we propose to use Raptor codes over $\GF(4)$, since the additive group of $\GF(4)$ is naturally matched to the QPSK constellation
generating group, and Raptor codes provide the necessary rate flexibility to adapt the system to the variations of the
source entropy rate, which may vary significantly depending on the symbol plane, the subband, and the image to be encoded.

The linear mapping from source to channel symbols allows to avoid the use of a conventional entropy coding stage, as in conventional baseline systems, and
this is expected to mitigate the catastrophic error propagation which affects conventional schemes in the presence of channel decoding residual errors.
The proposed JSCC scheme is able to achieve pure image compression performance almost identical to the state-of-the art.
While the proposed system for finite block length and transmission on the AWGN channel yields lightly worse bandwidth efficiency performance than the highly optimized
baseline system used by JPL in deep-space missions,  the results of Sec.~\ref{sec:Results} show that, as expected,
the new scheme has much improved robustness against  mismatched channel SNR conditions. While the baseline system requires the retransmission of a whole segment
in the presence of even one frame with residual post-decoding errors, the new scheme yields perceptual good image reconstruction quality
for SNR mismatch up to 1 dB below its nominal value, without any retransmission.

\appendix

\subsection{Isomorphism}\label{sec:app-iso}

In Sec. \ref{sec:code-design}, an isomorphism between JSCC and channel coding over the two-block composite channel has been shown under MAP decoding. In this section, the isomorphism is established under BP in the sense that at every iteration of the decoder, the set of messages generated by the message-passing BP decoder
for the source-channel coding scheme can be mapped into the corresponding set of messages generated by the message-passing BP decoder
for the associated channel coding scheme by a probability-preserving mapping. 

To this end, BP equations for a systematic Raptor code over $\GF(q)$ with $K$ input symbols and $K+n$ output symbols are given similar to \cite{Bennatan-nonbinary-journal}. The Tanner graph corresponding to the systematic Raptor code is given in Fig. (\ref{fig:DecRap2}). Let $\Gm = [\textbf{I};\Hm]$ denote the encoding matrix of the linear code formed by the Tanner graph of the systematic Raptor code. Then the codeword vector $\dv = [\uv \; \cv] = \uv\Hm$ has length $K+n$, where $\uv$ is the message vector.

Let us consider the $l^{\rm th}$ iteration of the BP decoder. BP messages can be represented in both probability and LLR domain as discussed earlier. In order to prove isomorphism, next we work in the probability domain where the messages are probability mass function (pmf) vectors of size $q$ with the following notation:

\begin{itemize}
\item $^{(l)}\mv_{\rm v , o}$ and $^{(l)}\mv_{\rm o, v}$
are the messages passed from the ${\rm v}^{\rm th}$ input node to the ${\rm o}^{\rm th}$ output node
and from the ${\rm o}^{\rm th}$ output node to the ${\rm v}^{\rm th}$ input node,
respectively, of the LT-decoder;
\item $^{(l)}\mv_{\rm v , c }$ and $^{(l)}\mv_{\rm c , v}$
are the messages passed from the ${\rm v}^{\rm th}$ variable node to the ${\rm c}^{\rm th}$ check node
and from the ${\rm c}^{\rm th}$ check node to the ${\rm v}^{\rm th}$ variable node, respectively, of
the LDPC decoder;
\item ${ \deltav}_{\rm ldpc}^{(l),{\rm v}}$ is the message generated from the ${\rm v}^{\rm th}$
LDPC variable node and passed to the corresponding input node of
the LT-decoder;
\item ${ \deltav}_{\rm lt}^{(l),{\rm v}}$ is the message generated from the ${\rm v}^{\rm th}$ LT input node
and passed to the corresponding variable node of the LDPC decoder; and
\item $\tv_{\rm o}$  is the input message to the BP decoder at the ${\rm o}^{\rm th}$ output node. This can be either the a-priori source probability or the posterior symbol-by-symbol probability given the channel outputs.

In case of a joint source-channel coding scheme the a-priori information on the first $K$ symbols are obtained from source statistics. Assume the source symbols, $u_{\rm v}$'s are i.i.d. selected from $\GF(q)$ with $Pr\{u_{\rm v} = g\} = P_U(g)$ for all $g\in GF(q)$ and $1\leq {\rm v}\leq K$. Let $\textbf{P}_U\eqdef [P_U(0),\ldots,P_U(q-1)] $ be a vector of size $q$ representing the pmf vector for the source distribution. Since the linear encoder is systematic

 \begin{equation}\label{eq:sourceZ}
\tv_{\rm o}= \textbf{P}_U.
\end{equation}

On the other hand, for the two-block composite channel scheme, the $\tv_{\rm o}$ for the first $K$ symbols are calculated using the discrete channel transition probability of the $\GF(q)$ additive noise channel. In this model, the transmitted vector $\xv$ is equal to the codeword $\dv$ and the received vector is $\yv = \xv-\zv$ where the operation is in $\GF(q)$. The noise symbols $z_{\rm o}$'s are i.i.d. selected from $\GF(q)$ with $Pr\{z_{\rm o} = g\} = P_U(g)$ for all $g\in GF(q)$ and $1\leq {\rm o}\leq K$ where the noise distribution is the same as the source distribution $\Pm_U$. Then,

\begin{equation}\label{eq:channelZ_GF}
t_{{\rm o},k}=Pr\{y_{\rm o}|x_{\rm o}=k\}=P_U(y_{\rm o}-k)=P_U(d_{\rm o}-z_{\rm o}-k).
\end{equation}

For both the JSCC scheme and the composite block channel scheme, the codeword symbols $d_{\rm o}$ for $K+1\leq o\leq K+n$ are first mapped to points in the constellation signal set $\mathfrak{X}$ by $\mu(c_{\rm o})$. Hence $y_{\rm o} = \mu(d_{\rm o})+z_{\rm o}$ is received at ${\rm o}^{\rm th}$ output node where $z_{\rm o}\in \Cc\Nc(0,N_0)$. Let the pdf of the  complex circularly symmetric AWGN noise is denoted by $f_z$. Then
\begin{eqnarray}
t_{{\rm o},k}&=&Pr\{y_{\rm o}|\mu(d_{\rm o})=\mu(k)\}\nonumber\\
\label{eq:t_awgn}
&=&\frac{f_z\left(||\mu(d_{\rm o})+z_{\rm o}-\mu(k)||\right)}{\sum_{k'}f_z\left(||\mu(d_{\rm o})+z_{\rm o}-\mu(k')||\right)}
\end{eqnarray}

Note that in AWGN initial channel message depends only on the distance between the observed vector and the hypotheses vector, hence we first focus on the calculation of the distance term, $||\mu(d_{\rm o})+z_{\rm o}-\mu(k)||$.
\begin{eqnarray*}
||\mu(d_{\rm o})+z_{\rm o}-\mu(k)||&=&||w_{d_{\rm o}}\left(\mu(0)\right)+z_{\rm o}-\mu(k)||\\
&=&||\mu(0)+w_{-d_{\rm o}}\left(z_{\rm o}\right)-w_{-d_{\rm o}}\left(\mu(k)\right)||\\
&=&||\mu(0)+w_{-d_{\rm o}}\left(z_{\rm o}\right)-\mu\left(k-d_{\rm o}\right)||.
\end{eqnarray*}
Substituting back into (\ref{eq:t_awgn}), we obtain
\begin{equation}
\label{eq:channelZ_AWGN}
t_{{\rm o},k}=\frac{f_z\left(||\mu(0)+w_{-d_{\rm o}}\left(z_{\rm o}\right)-\mu\left(k-d_{\rm o}\right)||\right)}{\sum_{k'}f_z\left(||\mu(0)+w_{-d_{\rm o}}\left(z_{\rm o}\right)-\mu\left(k'-d_{\rm o}\right)||\right)}.
\end{equation}

Next we investigate the relationship between input vectors $\tv$ of two different scenarios which has been already discussed in Sec. \ref{sec:code-design} to be isomorphic to each other. For convenience, we name the two-block composite channel as Scheme A and JSCC as Scheme B. 

Scheme B: This case corresponds to the joint source-channel coding problem where the source vector $\uv$ whose pmf is given by $\Pm_U$ is encoded by $\Gm$ and $\dv = \uv\Gm$ and the non-systematic part of the codeword is transmitted through the AWGN channel where additive noise vector is $\zv$. Then the $_B\tv_{{\rm o}}$ for scheme B is given directly using (\ref{eq:sourceZ}) and (\ref{eq:channelZ_AWGN}):
\begin{eqnarray}
\label{eq:ZB}
_Bt_{{\rm o},k}&=&\left\{
              \begin{array}{ll}
                P_U(k), & \hbox{if}\; 1\leq \rm o\leq K\\
                \frac{f_z\left(||\mu(0)+w_{-d_{\rm o}}\left(z_{\rm o}\right)-\mu\left(k-d_{\rm o}\right)||\right)}{\sum_{k'}f_z\left(||\mu(0)+w_{-d_{\rm o}}\left(z_{\rm o}\right)-\mu\left(k'-d_{\rm o}\right)||\right)}, & \hbox{otherwise.}
              \end{array}
            \right.
\end{eqnarray}

Scheme A: This case corresponds to the composite-channel coding problem where all zero codeword ($\dv = \zerov$) is transmitted through the composite channel. For the additive $\GF(q)$ channel we pick the noise pmf as $\Pm_U$. The noise realization for this channel is taken as $\uv$ (the same as the source vector $\uv$ used in Scheme B) and for the AWGN part the noise vector components are $w_{-d_{\rm o}}(z_{\rm o})$ for the ${\rm o}^{\rm th}$ output for $K+1\leq \rm o \leq K+n$. Notice that pdf of the complex circularly symmetric AWGN noise in Scheme B and Scheme A are the same since, the transformation is either a rotation or a reflection.

Then the $_At_{{\rm o},k}$ for scheme A is given as follows directly using (\ref{eq:channelZ_GF}) and (\ref{eq:channelZ_AWGN}), 
\begin{eqnarray}
\label{eq:ZA}
_At_{{\rm o},k}&=&\left\{
              \begin{array}{ll}
                P_U(u_{\rm o}+k), & \hbox{if}\; 1\leq \rm o\leq K\\
                \frac{f_{Z}\left(||\mu(0)+w_{-d_{\rm o}}(z_{\rm o})-\mu(k)||\right)}{\sum_{k'}f_{Z}\left(||\mu(0)+w_{-d_{\rm o}}(z_{\rm o})-\mu(k')||\right)}, & \hbox{otherwise.}
              \end{array}
            \right.
\end{eqnarray}

Next, we want to relate $_A\tv$ and $_B\tv$ in order to derive the relationship between BP messages in schemes A and B in Theorem \ref{theorem:Isomorphism_LT}. To this end, we define a shift operation $^+$ on pmf vectors as follows \cite{Bennatan-nonbinary-journal}: Let $g$ be an element of $\GF(q)$ and $\mv =[m_0,m_1,\ldots,m_{q-1}]$ be a pmf vector of size $q$ where the indices $i= 0,\ldots, q-1$ of each vector component are also interpreted as elements of $\GF(q)$. Index $i$ denotes the $i^{\rm th}$ element of $\GF(q)$ given some enumeration of the field elements where indices $0$ and $1$ are reserved for the zero and one elements of the field, respectively. Then $\mv^{+g} \eqdef [m_g,m_{g+1},\ldots,m_{q-1+g}]$  where summation is in the field. Then comparing (\ref{eq:ZB}) and (\ref{eq:ZA}), it is immediate to see that
\begin{equation}\label{eq:condition}
_A\tv^{-d_{\rm o}} = _B\tv,\;  1\leq \rm o \leq K+n.
\end{equation}
\end{itemize}
In the following theorem, we will prove that for any two schemes where $\tv$'s are related by (\ref{eq:condition}), the BP messages are related by (\ref{eqtn:theorem-iso_vo}) and (\ref{eqtn:theorem-iso_ov}).

{\begin{theorem}
\label{theorem:Isomorphism_LT}
\textit{Assume the input probability message vectors of two different schemes, A and B are related as follows:
\begin{equation}\label{eq:SchemeAB_relation}
_A\tv_{\rm o}^{-d_{\rm o}}= _B\tv_{\rm o},
\end{equation} where $d_{\rm o}$ is the value of the transmitted codeword $\dv=\uv \Gm$ for scheme B at the ${\rm o}^{\rm th}$ location. Then at any round $l$, the relationship between the messages passed in schemes A and B are as follows:
\begin{eqnarray}
\label{eqtn:theorem-iso_vo}
^{(l)}\bv_{{\rm ov}}& =& ^{(l)}\av_{{\rm ov}}^{-u_{\rm v}}\\
\label{eqtn:theorem-iso_ov}
^{(l+1)}\bv_{{\rm vo}}&=&^{(l+1)}\av_{{\rm vo}}^{-u_{\rm v}},
\end{eqnarray}
(where $\av$ is used to denote messages for Scheme-A and $\bv$ is used for Scheme-B) and $u_{\rm v}$ is the value of the ${\rm v}^{th}$ variable node at scheme B.}
\end{theorem}
}

Before proving the theorem, next we introduce some useful notation and functions that will be helpful to manipulate probability domain BP equations. The BP equations in the probability domain can be written in a compact way using vector shift operations and Discrete Fourier Transform (DFT) as done in \cite{Bennatan-nonbinary-journal}. One of these shift operations is $^\times$ of \cite{Bennatan-nonbinary-journal} which is similar to the previously defined $^+$ where summation is changed with multiplication in the field, i.e. $\mv^{\times g} \eqdef [m_0,m_{g},m_{2g},\ldots,m_{(q-1)g}]$. The following properties from \cite{Bennatan-nonbinary-journal} are useful for our derivations:
\begin{eqnarray*}
\left(\mv^{+g}\right)^{-g}&=&\mv,\;\mbox{and}\;\left(\mv^{\times g}\right)^{\times g^{-1}}=\mv,\;\;\mbox{if} \; g\neq 0\nonumber\\
\label{eq:plus_times}
\left(\mv^{+i}\right)^{\times g}&=&\left(\mv^{\times g}\right)^{+ig^{-1}}\\
\label{eq:times_plus}
\left(\mv^{\times g}\right)^{+ i}&=&\left(\mv^{+ gi}\right)^{\times g}
\end{eqnarray*}

In the following we use $\ev^T\dv$ to denote the scalar product of two vectors while $\ev\cdot\dv$ denotes componentwise multiplication which results in a vector. Similarly $\prodc$ denotes componentwise multiplications of multiple vectors. Additive vector representation \footnote{Note that finite fields exist for values of $q$ equal to $\p^\rr$ where $\p$ is a prime number and $\rr$ is a positive integer. Each element of $\GF(\p^\rr)$ can be represented as an $\rr$-dimensional vector over $\{0,\ldots,\p-1\}^\rr$. The sum of two $\GF(\p^\rr)$ elements corresponds to the sum of the vectors, evaluated as the modulo-$\p$ sum of vector components. This is called the additive vector-space representation.} of $g$ is denoted by the $\rr$-dimensional vector $\gu$. $\rr$-dimensional DFT and IDFT operations \cite{Bennatan-nonbinary-journal}, \cite{multidim_SP} for vectors of size $q = \p^\rr$ is described using the DFT pair $\fv, \mbox{DFT}(\fv)=\dv$:
\begin{eqnarray*}\label{eq:p^r_DFT}
d_{g} &=& \sum_{\hu\in \{0,\ldots,\p-1\}^\rr}
f_{h}e^{j(2\pi/\p)\hu^T\gu},\;\; 0\leq g\leq q-1,\\
\label{eq:p^r_IDFT}
f_{h}  &=&\frac{1}{q}\sum_{\gu\in \{0,\ldots,\p-1\}^r}
d_{g}e^{-j(2\pi/\p)\hu^T\gu},\;\; 0\leq h\leq q-1.
\end{eqnarray*}

We also define a new function $\Gamma(\cdot)$ and prove its properties which will be useful later in the proof of the theorem. Let $\Gamma(\cdot)$ be a function from $\GF(q)$ to complex vectors of size $q$ where the $h^{\rm th}$ component of the resultant vector is given as follows:

\begin{equation}\label{eq:gamma_def}
\left[\Gamma(g)\right]_h \eqdef e^{-\frac{j2\pi}{\p}\gu^T\hu}, \; 0\leq h\leq q-1.
\end{equation}

Then it is easy to show the following properties of $\Gamma(\cdot)$ function:
\begin{eqnarray}\label{eq:gamma_P1}
\mbox{DFT}\left(\ev^{+g}\right) &=& \mbox{DFT}\left(\ev\right)\cdot\Gamma(g)\\
\label{eq:gamma_P2}
\left[\mbox{IDFT}(\dv)\right]^{+g}&=&\mbox{IDFT}\left(\dv\cdot \Gamma(g)\right)\\
\label{eq:gamma_P3}
\Gamma(g)\Gamma(h) &=& \Gamma(g+h)
\end{eqnarray}

We let $g_{\rm ov} = g_{\rm vo}\in\GF(q)$ denote the value of the edge between the nodes $\rm o-\rm v$ and define $\N(\rm v)$  as the set of output $\rm o'$ nodes adjacent to the node $\rm v$. Note that we use the same $\N(\rm v)$ notation to denote the set of check nodes adjacent to $\rm v$ for LDPC part equations. It will be self-evident from the equations which set is considered. Similar neighbor notation $\N(\cdot)$ is also used for other node types. Using the notation above, the updating rules for
the LT and the LDPC decoders for the $l^{\rm th}$ iteration are given as follows:
\begin{eqnarray}
\label{eq:m_ov}
^{(l)}\mv_{{\rm ov}}^{\times -g_{{\rm ov}}^{-1}} & = & \mbox{IDFT}\left(\prodc_{{\rm v'\in \N(\rm o): \rm v'\neq \rm v}}\mbox{DFT}\left(^{(l)}\mv_{{\rm v'o}}^{\times g_{{\rm v'o}}^{-1}}\right)\cdot\mbox{DFT}\left(\tv_{\rm o}^{\times (-1)^{-1}}\right)\right),\\
\label{eq:m_vo}
^{(l+1)}\mv_{{\rm vo}} & =&  \left\{\begin{array}{ll}
                    \left(1/q\right)\onev, & \hbox{if}\; l=0; \\
                  \Uc\left[\deltav_{{\rm ldpc}}^{(l),{\rm v}}\cdot \prodc_{{\rm o'\in \N(\rm v): \rm o' \neq \rm o}}\,^{(l)}\mv_{{\rm o'v}}\right] &\hbox{otherwise},
                  \end{array}
                \right.
\end{eqnarray}
where the operator $\Uc$ normalizes the vector, in other words $\Uc\left[\mv\right] = \mv/\left(\onev^T\mv\right)$.
\begin{eqnarray*}
\label{eq:m_vc}
^{(l)}\mv_{{\rm vc}}& =&\left\{
                  \begin{array}{ll}
                    \left(1/q\right)\onev, & \hbox{if}\; l=0; \\
\Uc\left[\deltav_{{\rm lt}}^{(l),{\rm v}}\cdot\prodc_{{\rm c'\in\N(v): c'\neq c}}\,^{(l-1)}\mv_{{\rm c'v}}\right] &\hbox{otherwise,}
                                   \end{array}
                \right.\\
\label{eq:m_cv}
^{(l)}\mv_{{\rm cv}}^{\times -g_{{\rm cv}}^{-1}}  &= & \mbox{IDFT}\left(\prodc_{{\rm v'\in \N(\rm c): \rm v'\neq \rm v}}\mbox{DFT}\left(^{(l)}\mv_{{\rm v'c}}^{\times g_{{\rm v'c}}^{-1}}\right)\right).
\end{eqnarray*}
The messages $\deltav_{\rm lt}^{(l),{\rm v}}$ and $\deltav_{\rm
ldpc}^{(l),{\rm v}}$ passed from the LT to the LDPC decoder and from the
LDPC to the LT-decoder respectively are defined by:
\begin{equation*}
\label{eqtn:LT2LDPC} \deltav_{\rm lt}^{(l),{\rm v}}= \Uc\left[\prodc_{ {\rm o\in\N(v)}}\,^{(l)}\mv_{{\rm ov}}\right],\;
\deltav_{\rm ldpc}^{(l),\rm v}=\Uc\left[\prodc_{{\rm c\in \N(v)}}\,^{(l)}\mv_{{\rm cv}}\right] .
\end{equation*}

{\begin{proof}
The proof is based on induction on $(l)$. The relationship between the messages corresponding to different schemes for the $0^{\rm th}$ round, is verified first. Then round $(l+1)$ will be proven assuming the hypotheses for $l^{\rm th}$ round.

Round 0:

Using (\ref{eq:m_ov}-\ref{eq:m_cv}) it is immediate to see that
\begin{equation*}
^{(0)}\av_{{\rm ov}} = ^{(0)}\bv_{{\rm ov}} = ^{(1)}\av_{{\rm vo}} = ^{(1)}\bv_{{\rm vo}} =\left(1/q\right)\onev.
\end{equation*} Hence (\ref{eqtn:theorem-iso_vo}) and (\ref{eqtn:theorem-iso_ov}) are verified for $l = 0$. Now let's assume the theorem is true for $l^{\rm th}$ round, and prove it for $(l+1)^{\rm th}$ round using (\ref{eqtn:theorem-iso_vo}) and (\ref{eqtn:theorem-iso_ov}).

Round $l+1$:
\begin{eqnarray}
&&^{(l+1)}\bv_{{\rm ov}}^{\times -g_{{\rm ov}}^{-1}}  =  \mbox{IDFT}\left(\prodc_{{\rm v'\in \N(o): v'\neq v}}\mbox{DFT}\left(^{(l+1)}\bv_{{\rm v'o}}^{\times g_{{\rm v'o}}^{-1}}\right)\cdot\mbox{DFT}\left(_B\tv_{\rm o}^{\times (-1)^{-1}}\right)\right)\nonumber\\
&& \Ea  \mbox{IDFT}\left(\prodc_{{\rm v'\in \N(o): v'\neq v}}\mbox{DFT}\left(\left(^{(l+1)}\av_{{\rm v'o}}^{u_{\rm -v'}}\right)^{\times g_{{\rm v'o}}^{-1}}\right)\cdot\mbox{DFT}\left(\left(_A\tv_{\rm o}^{-d_{\rm o}}\right)^{\times (-1)^{-1}}\right)\right)\nonumber\\
&& \Eb  \mbox{IDFT}\left(\prodc_{{\rm v'\in \N(o): v'\neq v}}\mbox{DFT}\left(\left(^{(l+1)}\av_{{\rm v'o}}^{\times g_{{\rm v'o}}^{-1}}\right)^{-u_{\rm v'}g_{{\rm v'o}}}\right)\cdot\mbox{DFT}\left(\left(_A\tv_{\rm o}^{\times (-1)^{-1}}\right)^{-d_{\rm o} (-1)}\right)\right)\nonumber\\
&& \Ecc  \mbox{IDFT}\left(\prodc_{{\rm v'\in \N(o): v'\neq v}}\mbox{DFT}\left(^{(l+1)}\av_{{\rm v'o}}^{\times g_{{\rm v'o}}^{-1}}\right)\cdot\mbox{DFT}\left(_A\tv_{\rm o}^{\times (-1)^{-1}}\right) \cdot\Gamma(d_{\rm o})\cdot\prodc_{{\rm v'\in \N(o): v'\neq v}}\Gamma({-u_{\rm v}'g_{{\rm v'o}}})\right)\nonumber\\
&& \Ed  \mbox{IDFT}\left(\prodc_{{\rm v'\in \N(o): v'\neq v}}\mbox{DFT}\left(^{(l+1)}\av_{{\rm v'o}}^{\times g_{{\rm v'o}}^{-1}}\right)\cdot\mbox{DFT}\left(_A\tv_{\rm o}^{\times (-1)^{-1}}\right) \cdot\Gamma\left(d_{\rm o}+\sum_{{\rm v'\in \N(o): v'\neq v}}{-u_{\rm v}'g_{{\rm v'o}}}\right)\right)\nonumber\\
&& \Ee  \mbox{IDFT}\left(\prodc_{{\rm v'\in \N(o): v'\neq v}}\mbox{DFT}\left(^{(l+1)}\av_{{\rm v'o}}^{\times g_{{\rm v'o}}^{-1}}\right)\cdot\mbox{DFT}\left(_A\tv_{\rm o}^{\times (-1)^{-1}}\right) \cdot\Gamma\left(u_{\rm v}g_{\rm vo}\right)\right)\nonumber\\
&& \Ef  \mbox{IDFT}\left(\prodc_{{\rm v'\in \N(o): v'\neq v}}\mbox{DFT}\left(^{(l+1)}\av_{\rm v'o}^{\times g_{\rm v'o}^{-1}}\right)\cdot\mbox{DFT}\left(_A\tv_{\rm o}^{\times (-1)^{-1}}\right)\right)^{+u_{\rm v}g_{\rm vo}}\nonumber\\
&&\Eg\left(^{(l+1)}\av_{\rm ov}^{\times -g_{\rm ov}^{-1}} \right)^{+u_{\rm v}g_{\rm vo}}\nonumber\\
&&\Eh\left(^{(l+1)}\av_{\rm ov}^{ -u_{\rm v}} \right)^{\times-g_{\rm ov}^{-1}}\nonumber\\
\label{eq:proof_ov}
&&^{(l+1)}\bv_{\rm ov}=^{(l+1)}\av_{\rm ov}^{ -u_{\rm v}}
\end{eqnarray}
where $(a)$ is due to $l^{\rm th}$ round assumption (\ref{eqtn:theorem-iso_ov}); $(b), (h)$ are due to (\ref{eq:plus_times}), (\ref{eq:times_plus}). Properties of the $\Gamma(.)$ function, namely (\ref{eq:gamma_P1}-\ref{eq:gamma_P3}) are used to derive steps $(f), (c)$ and $(d)$. Step (g) is simply the corresponding BP equation (\ref{eq:m_ov}) for Scheme A. Lastly $(e)$ is due to the check constraint at the $\rm o^{\rm th}$ output node.

Assume the relationship for LDPC part messages, namely $^{(l)}\bv_{\rm cv},^{(l)}\av_{\rm cv},^{(l+1)}\bv_{\rm vc},^{(l+1)}\av_{\rm vc}$ similar to (\ref{eqtn:theorem-iso_vo}), (\ref{eqtn:theorem-iso_ov}) is already given. Then (\ref{eq:ldpc_part_assumption}) can be directly verified. For space concerns the assumption on the LDPC messages is not proven since it can be actually easily done using a similar proof to the current one.
\begin{equation}\label{eq:ldpc_part_assumption}
_B\deltav^{(l+1),{\rm v}}_{{\rm ldpc}}=\left(_A\deltav^{(l+1),{\rm v}}_{{\rm ldpc}}\right)^{-u_{\rm v}}
\end{equation}

Then using (\ref{eq:m_vo}), we write:
\begin{eqnarray}
^{(l+2)}\bv_{{\rm vo}} & =& \Uc\left[_B\deltav_{{\rm ldpc}}^{(l+1),{\rm v}}\cdot\prodc_{\rm o' \in \N(v):o'\neq o}\,^{(l+1)}\bv_{{\rm o'v}}\right],\nonumber\\
& =& \Uc\left[\left(_A\deltav_{{\rm ldpc}}^{(l+1),{\rm v}}\right)^{-u_{\rm v}}\cdot\prodc_{\rm o'\neq o}\,\left(^{(l+1)}\av_{{\rm o'v}}\right)^{-u_{\rm v}}\right],\nonumber\\
\label{eq:proof_vo}
^{(l+2)}\bv_{\rm vo} &=&^{(l+2)}\av_{\rm vo}^{-u_{\rm v}}.
\end{eqnarray}
(\ref{eq:proof_ov}) and (\ref{eq:proof_vo}) completes the proof.
\end{proof}}

 Due to (\ref{eq:condition}), according to Theorem (\ref{theorem:Isomorphism_LT}), the BP messages of Scheme A and Scheme B can be obtained from each other using the operation $^+$. Hence it can be easily seen that the error probability of Scheme B with $\uv$ and $\zv$ is equal to the error probability of Scheme A with $\uv$ and $w_{-d_{\rm o}}(z_{\rm o})$. Since AWGN is isomorphic, the average error probability of Scheme B (JSCC) is equal to the average error probability of the Scheme A (composite two-block channel).

\subsection{Stability Condition}\label{sec:app-stability}

In this section we extend the stability condition of Etesami {\em et al.} \cite{EtSho06}
to LT codes over $\GF(4)$ and the two-blocks composite channel of Definition \ref{composite-channel}
with parameters $H$ and $C$.  Let $\Fm$ be the $4 \times 4$ DFT with elements $[\Fm]_{m,\ell} = e^{-j\pi (m-1)(\ell-1)/2}$ for
$m,\ell \in \{1,2,3,4\}$.
The BP messages for $q$-ary codes can be either represented as probability vectors  (of length $q$) or
as LLR  vectors of length $q-1$. Let $\LLR : \mv \mapsto \Lm$ denote the mapping of the probability representation into the LLR representation from the probability domain.
We define the mapping $\Phi : \RR^3 \rightarrow [0,1]^4$ given by
\[ \Phi(\Lm) \eqdef \Fm \; \LLR^{-1} ( \Lm). \]
Under the Gaussian approximation $\Lm \sim \Nc(\upsilon \onev,\Sigma_\upsilon)$ of  \cite{Bennatan-nonbinary-journal},
already used in Sec.~\ref{sec:EXIT}, we have
\begin{equation}\label{eq:E_Phi}
\mathbb{E}\left[\Phi(\Lm)\right] = \left[1,\;\Psi(\upsilon),\;\Psi(\upsilon),\;\Psi(\upsilon)\right],
\end{equation}
with
\begin{equation}\label{eq:psi_nu}
\Psi(\upsilon)\eqdef \mathbb{E}\left[\frac{1-e^{-L_1}+e^{-L_2}-e^{-L_3}}{1+e^{-L_1}+e^{-L_2}+e^{-L_3}}\right].
\end{equation}
Since the right hand side of (\ref{eq:E_Phi}) depends only on a scalar value,
when we consider the expectation of LLR values under Gaussian approximation, it will suffice to work with a ``scalar'' version of the expectation
operator, denoted by $\Es$. For example, $\Es\left[\Lm\right] = \upsilon$ is the short-hand notation for $\EE[\Lm] = \upsilon\onev$.
With this notation, we have  $\Es\left[\Phi(\Lm)\right] = \Psi(\upsilon) = \Psi\left(\Es\left[\Lm\right]\right)$.

The BP message updating equations in the LLR domain, assuming an output node o with neighborhood $\N({\rm o})$ of size
$|\N({\rm o})| = i$ and an input node v with neighborhood $\N({\rm v})$ of size $|\N({\rm v})|= j$, are given by
\begin{eqnarray*}\label{eq:LLR domain_with_phi}
\Phi\left(\Lm^{(l)}_{\rm o,v}\right) & = & \left[\prod_{\imath=1}^{i - 1} \Phi\left(\Lm^{(l)}_{{\rm v}_\imath,{\rm o}} \right)\right] \cdot \Phi(\tv_{\rm o}) \nonumber\\
\Lm^{(l+1)}_{\rm v,o} & =& \sum_{\jmath = 1}^{j-1}\,\Lm^{(l)}_{{\rm o}_\jmath, {\rm v}}.
\end{eqnarray*}
where $\tv_{\rm o}$ denotes the LLR of the channel output for node o, and $l$ denotes the BP iteration.

Following \cite{EtSho06},  we are interested in the evolution of the quantity $\Es\left[\Lm^{(l)}_{\rm o,v}\right]$ in a right neighborhood of $0$.
We can write
\begin{eqnarray*}
&&\Es\left[\Lm^{(l)}_{\rm o,v}\right]=\sum_i \omega_i \Es\left[\Lm^{(l)}_{\rm o,v}\mid |\N({\rm o})|=i \right]\\
&&=\sum_i \omega_i \Psi^{-1} \left ( \Psi\left(\Es\left[\Lm^{(l)}_{\rm o,v}\mid |\N({\rm o})|=i\right]\right)\right )\\
&&=\sum_i \omega_i \Psi^{-1}\left ( \Es\left[\Phi\left(\Lm^{(l)}_{\rm o,v}\mid |\N({\rm o})|=i\right)\right]\right )\\
&&\Ea\sum_i \omega_i \Psi^{-1}\Big\{\left[\gamma\Psi\left(J^{-1}(1-H)\right)+(1-\gamma)\Psi\left(J^{-1}(C)\right)\right]
\left[  \Es\left[\Phi\left(\Lm^{(l)}_{\rm v,o}\right)\right] \right ]^{i-1}\Big\}\\
&& = \sum_i \omega_i \Psi^{-1}\left ( \left[ \gamma \Psi\left(J^{-1}(1-H)\right) + (1-\gamma) \Psi\left(J^{-1}(C) \right)\right] \cdot
\left[ \sum_j \iota_j \Psi\left((j-1)\Es\left[\Lm^{(l-1)}_{\rm o,v} \right]\right) \right ]^{i-1}\right ),
\end{eqnarray*}
where in (a) we used the two-block composite channel property. For successful start of the decoding under the Gaussian approximation, the quantity
$\Es\left[\Lm^{(l)}_{\rm o,v}\right]$ must be strictly increasing from one iteration to the other in a sufficiently small right neighborhood
of zero. A necessary condition is that
\begin{eqnarray}
\upsilon & < &\sum_i \omega_i \Psi^{-1} \left ( \left [  \gamma\Psi\left(J^{-1}(1-H)\right)+(1-\gamma)\Psi\left(J^{-1}(C)\right) \right] \left [
\sum_j \iota_j \Psi\left((j-1) \upsilon \right)\right ]^{i-1} \right ) \label{stability1}
\end{eqnarray}
in a sufficiently small right neighborhood of $\upsilon = 0$. By taking derivative of both sides of (\ref{stability1}) with respect to $\upsilon$ at $0$ and using
$\Psi(0) = 0$ and $\Psi'(0) \neq 0$ (see at the end of this section), after some algebra we arrive at the stability condition:
\begin{equation}\label{eq:stability}
\Omega_2 \geq \Upsilon\left(\gamma, r_{\rm ldpc}, H,C\right) \eqdef \frac{\gamma/r_{\rm ldpc}}{2\left(\gamma \Psi\left(J^{-1}(1-H)\right)+(1-\gamma)\Psi\left(J^{-1}(C)\right)\right)}\;.
\end{equation}
%
%
\textbf{Calculation of $\Psi'(0)$.}
We rewrite $\Psi(\upsilon)$ in (\ref{eq:psi_nu}) using the zero mean Gaussian random vector $\mathbf{\Lc}\sim \Nc( \zerov,\Sigma_\upsilon)$ as
\begin{equation*}
\Psi(\upsilon)= \mathbb{E}\left[\frac{1-e^{-\Lc_1-\upsilon}+e^{-\Lc_2-\upsilon}-e^{-\Lc_3-\upsilon}}{1+e^{-\Lc_1-\upsilon}+e^{-\Lc_2-\upsilon}+e^{-\Lc_3-\upsilon}}\right].
\end{equation*}
Using a Taylor expansion in a neighborhood of $\upsilon = 0$, we obtain
\begin{eqnarray}
\label{eq:psi_nu_n1n2n3}
\Psi(\upsilon)&=&\sum_{n_1=0}^\infty\sum_{n_2=0}^\infty\sum_{n_3=0}^\infty \frac{\mathbb{E}\left[ \Lc_1^{n_1}\Lc_2^{n_2}\Lc_3^{n_3}\right]}{n_1!n_2!n_3!}h_{n_1,n_2,n_3}(\upsilon),
\end{eqnarray}
where $\eta\eqdef n_1+n_2+n_3$ and
\begin{equation}
\label{eq:h}
h_{n_1,n_2,n_3}(\upsilon) \eqdef \left . \frac{\partial^{\eta}\left(\frac{1-e^{-l_1-\upsilon}+e^{-l_2-\upsilon}-e^{-l_3-\upsilon}}{1+e^{-l_1-\upsilon}+e^{-l_2-\upsilon}+e^{-l_3-\upsilon}}\right)}{\partial l_1^{n_1}\partial l_2^{n_2}\partial l_3^{n_3}} \right |_{l_1=l_2=l_3=0}.
\end{equation}
Let $\bar \Lm$ be a zero mean Gaussian random vector of size $\eta$
with covariance matrix $\Kc = [\Kc_{i,j}]$, and let $\Xi_\eta$ denote the collection of all unordered sequences
of all unordered integer pairs, with each sequence containing each integer from $1$ to $\eta$ exactly once.
Then, the well-known formula for the higher moments of real Gaussian random variable yields,
\begin{eqnarray}
\label{eq:higher-gauss-moment}
\mathbb{E}\left[ \prod_{t=1}^\eta \bar{L}_t\right]&=&\left\{
                                                        \begin{array}{ll}
                                                          0, & \hbox{for $\eta$ odd} \\
                                                          \sum_{\left(i_1,j_1,\ldots,i_{\eta/2},j_{\eta/2}\right)\in \Xi_\eta} \prod_{s=1}^{\eta/2} \Kc_{i_s,j_s}, & \hbox{for $\eta$ even.}
                                                        \end{array}
                                                      \right.
\end{eqnarray}
Specializing the components of $\bar{\Lm}$ to be $\bar{L}_m = \Lc_1$ for $1\leq m \leq n_1$,
$\bar{L}_m = \Lc_2$ for $n_1+1\leq m \leq n_1+n_2$ and, $\bar{L}_m = \Lc_3$ for $n_1+n_2+1\leq m \leq \eta$,
using the form of $\Sigma_\upsilon$ as given in Sec. \ref{sec:code-design}, and using (\ref{eq:higher-gauss-moment}),
it is not difficult to see that $\mathbb{E}\left[ \Lc_1^{n_1}\Lc_2^{n_2}\Lc_3^{n_3}\right]\propto |\Xi_\eta| \upsilon^{\eta/2} $
for all even values of $\eta$, and zero for all odd values of $\eta$. Then, we have
\begin{eqnarray*}
\Psi(\upsilon) & = &\sum_{n_1,n_2,n_3\; : \; \eta \;\mbox{\tiny is even} } \frac{ |\Xi_\eta| \upsilon^{\eta/2}}{n_1!n_2!n_3!} h_{n_1,n_2,n_3}(\upsilon)\\
\Psi'(\upsilon)& = &\sum_{n_1,n_2,n_3\; : \; \eta \;\mbox{\tiny is even} } |\Xi_\eta|\frac{\frac{\partial h_{n_1,n_2,n_3}(\upsilon)}{\partial \upsilon}\upsilon^{\eta/2}+\frac{\eta}{2}\upsilon^{\eta/2-1}h_{n_1,n_2,n_3}(\upsilon)}{n_1!n_2!n_3!}
\end{eqnarray*}
For $\upsilon = 0$, the above summations includes the triplets $(n_1, n_2,n_3)$ with sum $\eta = n_1 + n_2 + n_3 = 2$. These are
$(2 0 0),(0 2 0),(0 0 2),(1 1 0),(0 1 1)$, and $(1 0 1)$. For any $(n_1,n_2,n_3)$ in this set, we have
$\frac{\mathbb{E}\left[ \Lc_1^{n_1}\Lc_2^{n_2}\Lc_3^{n_3}\right]}{n_1!n_2!n_3!} = \upsilon$, since the diagonal elements of
$\Sigma_\upsilon$ are equal to $2\upsilon$ and the off-diagonal elements are equal to $\upsilon$.
Then, we obtain
\begin{eqnarray*}
\Psi'(0)&= &\sum_{\eta=2} h_{n_1,n_2,n_3}(0)\\
&=&h_{2,0,0}(0)+h_{0,2,0}(0)+h_{0,0,2}(0)+h_{1,1,0}(0)+h_{0,1,1}(0)+h_{1,0,1}(0)\\
&=& -60/16,
\end{eqnarray*}
where, using (\ref{eq:h}), we have $h_{2,0,0}(0) = h_{0,0,2}(0) = -2,\;h_{1,1,0} = 0,\;h_{1,0,1} = 1/8,\;h_{0,1,1}(0) =  -1/16$ and
$h_{0,2,0}(0) = 3/16$.


\end{document}